\documentclass{article}
\usepackage{custom_tex}

\title{Rotated Mean-Field Variational Inference \\ and Iterative Gaussianization}

\author[1]{Yifan Chen}
\author[2]{Sifan Liu}
\affil[1]{Department of Mathematics, University of California, Los Angeles}
\affil[2]{Department of Statistical Science, Duke University}
\date{July 2026}

\begin{document}
\maketitle

\begin{abstract}
We propose an iterative Gaussianization method for sampling from unnormalized densities by repeatedly applying mean-field variational inference (MFVI) in rotated coordinate systems. At each iteration, the method selects a rotation, solves an MFVI subproblem in the rotated coordinates, and applies the inverse coordinatewise map to transform the current target closer to the standard Gaussian. The resulting algorithm provides a computationally efficient way to construct flow-like transport maps: it requires only MFVI subproblems, avoids large-scale optimization, and produces transformations that are easy to invert and evaluate.

The effectiveness of the procedure depends on selecting informative rotations. We develop an efficient PCA-type method that chooses rotations from the leading eigenvectors of a cross-covariance matrix involving the target's score function. Experiments on Bayesian posterior sampling tasks show that performing MFVI in the proposed PCA-rotated coordinate systems substantially improves over standard MFVI, and that the resulting iterative Gaussianization procedure provides accurate flow-like approximations at lower computational cost than \sloppy{conventional} normalizing-flow variational approximations.

\end{abstract}

\section{Introduction}

Sampling from an unnormalized density is a fundamental problem in statistics, with applications in Bayesian inference, inverse problems, statistical mechanics, and many others. Markov chain Monte Carlo (MCMC; \citep{metropolis1953equation,hastings1970monte}) provides a general-purpose framework for this task, but it can become computationally expensive when applied to large datasets and often suffers from slow mixing in high dimensions. The difficulties are compounded when the target distribution exhibits strong correlations, variable curvature, or multiple modes—features that frequently arise in hierarchical models and demand sophisticated algorithmic adaptation. Additionally, while ensemble methods enable some degree of parallelization, the sequential dependence inherent in Markov chains limits their scalability on modern hardware.

Variational inference (VI; \citep{jordan1999introduction,blei2017variational}) offers a scalable alternative by recasting sampling as an approximation problem: the goal is to find a distribution within a pre-specified variational family that is closest to the target under a divergence measure. A widely used variational family is the mean-field family, which approximates the target with a product distribution. Mean-field variational inference (MFVI) is conceptually simple and can be implemented efficiently using coordinate ascent algorithms (CAVI; \citep{bishop2006pattern}). MFVI has been extensively studied, including its consistency and convergence rates in statistical estimation~\citep{bickel2013asymptotic,zhang2020convergence}, applications to model selection~\citep{zhang2024bayesian}, algorithmic convergence of CAVI~\citep{arnese2024convergence,lavenant2024convergence,bhattacharya2025convergence}, alternative algorithms~\citep{tran2023particle,du2024particle,jiang2025algorithms}, and theoretical properties for log-concave targets~\citep{lacker2024mean}. Despite its computational appeal and theoretical understanding, MFVI is known to underestimate marginal variances and to yield over-confident uncertainty quantification, as it cannot capture dependencies among coordinates~\citep{neville2014mean}.

Beyond MFVI, richer variational families can be constructed to capture correlations, such as by using multivariate Gaussian or mixture distributions~\citep{opper2009variational,lin2019fast,che2025stable}. A more flexible approach, inspired by generative modeling, parametrizes the variational family as the pushforward of a simple reference distribution (e.g., a standard Gaussian) through an expressive class of transport maps. Normalizing flows~\citep{papamakarios2021normalizing} offer a general framework for this construction by composing multiple simple, invertible transformations. With sufficient capacity, flows achieve leading performance in density estimation and generative modeling. For sampling tasks, recent empirical studies show that flows can approximate complex targets with high accuracy, but only under demanding conditions: the model must be highly expressive, optimization requires very large Monte Carlo batches (e.g., $2^{19}$ samples), and the learning rate must be carefully tuned~\citep{rezende2015variational,blessing2024beyond,agrawal2025disentangling}. While normalizing flows offer a promising approach to sampling, the training cost remains prohibitively high for many practical applications.

Motivated by the tractability of MFVI and the compositional design of normalizing flows, we propose a simple but effective strategy: construct expressive transport maps by iteratively performing MFVI in different coordinate systems. The key insight is that while MFVI in any fixed coordinate system produces only a product distribution, alternating between MFVI steps and orthogonal rotations can capture complex dependencies. This raises a natural question: \emph{What is the best coordinate system in which to perform MFVI?}

Directly optimizing the rotation is computationally infeasible, since
each candidate rotation would require solving a full MFVI problem, and
the outer optimization would require searching over the orthogonal
group.
We therefore study a local analogue of the same question: starting from
the standard Gaussian, how fast can the KL divergence be decreased by an
infinitesimal coordinatewise transport in a given coordinate system?
This instantaneous analysis shows that the steepest KL decrease
rate is given by the \emph{projected Fisher information} of the target
relative to the standard Gaussian~\citep{lacker2023independent}, making it a natural criterion for
selecting rotations.
We then derive an ANOVA-type decomposition of the projected Fisher
information based on orthogonal Hermite polynomials.
The leading-order contribution in this decomposition is maximized by
aligning the rotation with the eigenvectors of the matrix
\begin{align*}
    H
    =
    \EE[\bfx\sim\N(0,I_d)]
    {
        \bfx
        \big(
            \nabla\log p(\bfx) + \bfx
        \big)\tran
    } ,
\end{align*}
which is the cross-covariance between
$\bfx\sim\N(0,I_d)$ and the relative score
$\nabla\log p(\bfx)+\bfx$.
We call the resulting rotation-selection method
\emph{relative score PCA}.

Performing MFVI in the coordinate system selected by relative score PCA produces a transport map that combines a rotation with a coordinatewise transformation.
Together, these operations transform the target closer to the standard Gaussian. Repeating this procedure gives an iterative scheme in which each step consists of selecting a rotation by relative score PCA and then applying MFVI in the rotated coordinates. We refer to this algorithm as \emph{iterative Gaussianization}, since it parallels Gaussianization methods for density estimation~\citep{chen2000gaussianization}, which alternate between rotations and marginal Gaussianization. In our setting, however, the target is available through its unnormalized density and score function rather than through exact samples, and the rotation is chosen by a score-based PCA criterion tailored to MFVI.

Performing MFVI on a rotated target can also be viewed as a form of \sloppy{\emph{model reparametrization}}, a long-standing strategy in Bayesian computation for improving the efficiency of inference algorithms. Classical examples include centered and non-centered parametrizations in hierarchical models~\citep{papaspiliopoulos2007general}. Recent work has extended this principle to variational inference, where suitable reparametrizations can substantially improve the quality of variational approximations~\citep{tan2021use}. In contrast to these model-specific approaches, the proposed rotation method provides a model-agnostic, score-based reparametrization. Relative score PCA is also related to PCA-based methods such as active subspaces~\citep{constantine2014active} and certified dimension reduction~\citep{zahm2022certified}, which seek informative low-dimensional subspaces for function approximation or likelihood-informed posterior approximation. Our objective is different: we seek rotations that make coordinatewise MFVI more effective; see Section~\ref{sec-comparison-to-active-subspace} for a detailed comparison.

In summary, the main contributions of this work are as follows:
\begin{enumerate}
    \item We introduce an iterative Gaussianization framework that alternates between rotations and MFVI updates, progressively transforming the target distribution toward the standard Gaussian. The inverse transformations define a flow-like variational approximation that is easy to sample from, invert, and evaluate, while avoiding the joint high-dimensional optimization required by conventional normalizing flows.

    \item We develop relative score PCA as an efficient rotation-selection method for enhancing MFVI. It requires only samples from the standard Gaussian and evaluations of the target score, and it adds little computational overhead to standard MFVI. In addition, for rotated product targets, relative score PCA recovers the optimal rotation when its associated eigenvalues are distinct (Proposition~\ref{prop: rotated product exact recovery}); in this case, the subsequent MFVI approximation is exact.

    \item Our theoretical analysis establishes a characterization of Gaussianity: if a distribution admits the standard Gaussian as its optimal mean-field approximation under any rotation, then the distribution itself must be standard Gaussian. This result identifies the standard Gaussian as the unique stationary point of the proposed iterative algorithm and provides a characterization of mean-field optimality under rotations that may be of independent interest.
\end{enumerate}

The remainder of the paper is organized as follows.
Section~\ref{sec: mfvi background} formulates MFVI as a
coordinatewise transport problem and introduces iterative
Gaussianization based on rotated MFVI.
Section~\ref{sec: stationarity guarantee} establishes a stationarity
guarantee for this iterative procedure.
Section~\ref{sec: instantaneous} shows that the instantaneous KL
decrease rate under coordinatewise transports is given by the projected
Fisher information, and develops a constrained and weighted Wasserstein gradient flow
interpretation of iterative Gaussianization.
Section~\ref{sec: pca} derives relative score PCA from an orthogonal
decomposition of projected Fisher information and connects it to Stein
discrepancy and other gradient-based PCA methods.
Section~\ref{sec: experiments} presents numerical experiments, and
Section~\ref{sec: conclusion} concludes.

\section{MFVI and iterative Gaussianization}
\label{sec: mfvi background}

Let $p$ denote the target distribution supported on $\R^d$, with Lebesgue density $p(\bfx)\propto \exp(-U(\bfx))$ assumed to be continuously differentiable and known up to the normalizing constant. For simplicity, and with a slight abuse of notation, we use the same symbol for a distribution and its density, and do not distinguish the two when the meaning is clear from context. Let $\pprod$ be the set of all product measures on $\R^d$ whose densities are positive and continuously differentiable. Define $\calF$ as the set of coordinatewise $C^2$-diffeomorphisms on $\R^d$, i.e. {$\calF=\{F: \R^d \to \R^d \mid F(\bfx) = (F_1(x_1), \ldots, F_d(x_d)), F_i:\R\to\R\ C^2\text{-diffeomorphism}\}$}. For a distribution $p$ and diffeomorphism $F$ on $\R^d$, the pushforward $F\# p$ denotes the distribution of $F(X)$ when $X\sim p$. Let $\gamma=\N(0,I_d)$ denote the standard Gaussian distribution on $\R^d$.

\subsection{Mean-field variational inference}

MFVI approximates the target $p$ with a product distribution that minimizes the KL divergence:
\begin{align}\label{equ: mfvi}
    q^* = \underset{q\in\pprod}{\argmin}\; \kl{q}{p},
\end{align}
where the Kullback--Leibler (KL) divergence is defined as $\kl{q}{p} = \EE[q] {\log \frac{q(\bfx)}{p(\bfx)}}$.
The solution $q^*$ of problem~\eqref{equ: mfvi} satisfies the first-order optimality condition
\begin{align}\label{equ: mfvi opt cond}
    \nabla_{i} \log q^*_i(x_i) = \EE[q^*_{-i}]{\nabla_{i} \log p(\bfx)}, \quad i=1,\ldots,d,
\end{align}
where $\bbE_{q^*_{-i}}$ denotes expectation under $q^*$ conditioned on $x_i$ and $\nabla_i$
denotes differentiation with respect to the $i$-th coordinate. For a general target this
identity is understood in a weak sense; see~\citep[Equation 4.7]{arnese2024convergence} for
a derivation assuming only that $q^*$ has a finite $m$-th moment ($m\geq2$) and
$|\log p(\bfx)|\leq c(1+\|\bfx\|_2^m)$, $\|\nabla\log p(\bfx)\|_2\leq c(1+\|\bfx\|_2^m)$ for
some $c>0$ and almost every $\bfx\in\R^d$. We assume these growth conditions throughout.
Together with our standing assumptions that $p\in C^1$ and $\pprod$ consists of positive $C^1$ product densities, \eqref{equ: mfvi opt cond} holds pointwise for every $x_i$.

In general, the solution of the mean-field equation~\eqref{equ: mfvi opt cond} need not be unique. If the potential function $U$ is convex, all solutions of~\eqref{equ: mfvi opt cond} are global minimizers of the MFVI problem~\eqref{equ: mfvi}. If $U$ is strongly convex, the global minimizer is unique~\citep[Theorem 1.1, Proposition 3.9]{lacker2024mean}. For the theoretical results in this paper, we assume that each MFVI subproblem has a finite attained minimizer; when minimizers are not unique, we fix a particular choice. In the numerical experiments, this exact update is approximated by parametrizing the coordinatewise diffeomorphism $F$ as monotone splines and optimizing their parameters with Adam; see Section~\ref{sec: experiments} for details.

The MFVI problem~\eqref{equ: mfvi} can equivalently be formulated as finding the best coordinatewise map that pushes the standard Gaussian $\gamma$ forward to $p$:
\begin{align*}
    F^*=\underset{F\in\calF}{\argmin}\; \kl{F\# \gamma}{p}.
\end{align*}
In fact the two formulations range over the same set of distributions. Given $q\in\pprod$
with marginal cumulative distribution functions (CDFs) $Q_1,\dots,Q_d$, the coordinatewise map $F$ with $F_i=Q_i^{-1}\circ\Phi$
(where $\Phi$ is the standard Gaussian CDF) belongs to $\calF$ and satisfies $F\#\gamma=q$;
here $F_i\in C^2$ precisely because $Q_i'=q_i$ is $C^1$ and positive. Conversely, for any
$F\in\calF$ the pushforward $F\#\gamma$ is a product measure with positive $C^1$ density, and
hence lies in $\pprod$. 

Once $F^*$ is obtained, we can Gaussianize the target by defining $p^* := F^{*-1}\# p$, which is closer to the standard Gaussian than the original target $p$, since
\begin{align*}
    \kl{\gamma}{p^*} = \kl{F^* \# \gamma}{p}\leq \kl{\gamma}{p}.
\end{align*}
The equality follows from the invariance of KL divergence under invertible transformations, and the inequality follows from the optimality of $F^*$.

\subsection{Iterative Gaussianization via rotated MFVI}
\label{sec: mfvi rotated}
Performing MFVI in the standard coordinate system can lead to poor approximation when the target distribution exhibits strong correlations among variables.
To introduce dependence among variables, we can perform MFVI in a rotated coordinate system. Given an orthogonal matrix $R\in\R^{d\times d}$ with $RR\tran=I_d$, we first rotate the target distribution to $p_R=R\# p$, whose density is $p_R(\bfx) = p(R\tran \bfx)$, and then perform MFVI for the rotated target. Specifically, we solve 
\begin{align*}
    F^* = \underset{F\in\calF}{\argmin}\; \kl{F\# \gamma}{p_R}.
\end{align*}
As before, once $F^*$ is found, the target can be transformed to $p_R^* := F^{*-1}\# p_R$, which is closer to the standard Gaussian than the original target $p$, since
\begin{align*}
    \kl{\gamma}{p_R^*} \leq \kl{\gamma}{p_R}=\kl{\gamma}{p},
\end{align*}
where the equality holds because KL divergence is invariant under orthogonal transformations and the standard Gaussian distribution $\gamma$ is rotationally invariant. 

Although performing MFVI in a rotated coordinate system captures some dependence among coordinates, its expressiveness remains limited: the resulting approximation is still a product distribution in a single rotated coordinate system. To obtain a richer approximation, we iterate the rotation-and-MFVI step, treating each Gaussianized output as the new target and progressively bringing it closer to the standard Gaussian.

Formally, let $p=p^{(0)}$ denote the target, and let $p^{(k-1)}$ denote the target after $k-1$ iterations ($k\geq1$). In the $k$-th iteration, we first choose an orthogonal matrix $R_k\in\R^{d\times d}$ and form the rotated target $p^{(k-1)}_{R_k} = R_k \# p^{(k-1)}$. We then solve the MFVI problem
\begin{align}\label{equ: kl decrease condition}
    F_k = \underset{F\in\calF}{\argmin}\; \kl{F \# \gamma}{p^{(k-1)}_{R_k}} .
\end{align}
Finally, we apply the inverse transformation to Gaussianize the rotated target, yielding the new target
$p^{(k)} = F_k^{-1} \# p^{(k-1)}_{R_k}$. 
This iterative procedure, which we call \emph{iterative Gaussianization},
is summarized in Algorithm~\ref{algo: iterative gaussianization}.

\begin{algorithm}
\caption{Iterative Gaussianization}
\label{algo: iterative gaussianization}
\begin{algorithmic}
    \REQUIRE {Target distribution $p$; number of iterations $K$}
    \STATE {Initialize $p^{(0)}=p$}
    \FOR {$k=1$ to $K$}
        \STATE {Select an orthogonal matrix $R_k$ and define $p_{R_k}^{(k-1)}= R_k\# p^{(k-1)}$} 
        \vspace{.4em}
        \STATE {Solve for $F_k=\underset{F\in\calF}{\argmin}\; \kl{F \# \gamma}{p_{R_k}^{(k-1)}}$}
        \vspace{.4em}
        \STATE {Update $p^{(k)}=F_k^{-1} \# p_{R_k}^{(k-1)} $}
    \ENDFOR
    \RETURN {Gaussianization transformation $ F_K^{-1}\circ R_K \circ \cdots \circ F_1^{-1}\circ R_1$}
\end{algorithmic}
\end{algorithm}

After $k$ iterations, the original target $p$ is transformed to
\begin{align*}
    p^{(k)} = (F_k^{-1} \circ R_k \circ \cdots \circ F_1^{-1} \circ R_1) \# p.
\end{align*}
Compositions of rotations and coordinatewise maps are known to be universal approximators in the density estimation setting \citep{chen2000gaussianization}: there exist sequences of rotations $R_k$ and coordinatewise maps $F_k$ such that $p^{(k)}$ converges weakly to $\gamma$ as $k\to\infty$. Thus, iterative Gaussianization has the expressiveness needed to approximate general targets, even though each iteration solves only a mean-field subproblem. The composite transformation also resembles a normalizing flow \citep{papamakarios2021normalizing}. The key difference is that our construction is built sequentially from MFVI subproblems, whereas normalizing flows typically optimize the full composition jointly.

The inverse of the sequence of transformations pushes the standard Gaussian toward the target distribution, resulting in the approximation
\begin{align}\label{equ: qk}
    q^{(k)} = (R_1\tran \circ F_1 \circ \ldots \circ R_k\tran \circ F_k) \# \gamma.
\end{align}
This can be used to generate approximate samples from the target distribution by pushing forward samples from the standard Gaussian. The Jacobian determinant of this transformation is the product of the determinants of each coordinatewise map, since rotations have unit determinant. Therefore, the density of $q^{(k)}$ can be evaluated efficiently, which can be used in a subsequent step to correct for the bias in $q^{(k)}$ through importance sampling or MCMC.

By construction, each iteration decreases the KL divergence to the
standard Gaussian, regardless of how the rotation $R_k$ is chosen.
This monotonicity is stated in the following proposition.
\begin{prop}
    \label{lem: monotonicity}
    The KL divergence is non-increasing in $k$:
    \begin{align*}
        &\kl{\gamma}{p^{(k)}} \leq \kl{\gamma}{p^{(k-1)}},\quad \kl{q^{(k)}}{p} \leq \kl{q^{(k-1)}}{p},\quad \forall\, k\geq 1.
    \end{align*}
\end{prop}

For Gaussian targets, we have a more explicit characterization of the contraction rate of the KL divergence, under random rotations. 
The proof of the following theorem, given in Appendix~\ref{sec: proof of gaussian convergence}, relies on the fact that the optimal mean-field approximation of each iterate $p^{(k)}$ ($k\geq1$) is $\gamma$ by construction.
\begin{thm}\label{thm: gaussian convergence}
    If the target distribution $p$ is Gaussian, then all iterates $p^{(k)}$ are Gaussian. Suppose that the condition number of the covariance matrix of $p^{(k)}$ ($k\geq1$) is $\chi$. Then
    \begin{align*}
        \EE{\kl{\gamma}{p^{(k+1)}} } \leq \Big(1-\frac{2}{(d+2)\chi^2}\Big) \cdot \kl{\gamma}{p^{(k)}},
    \end{align*}
    where the expectation on the left-hand side is conditional on the current iterate $p^{(k)}$ and taken over the rotation $R_{k+1}$, chosen uniformly at random from the orthogonal group $O(d)$.
\end{thm}

\subsection{Related work}
\label{sec: related work}

\paragraph{Gaussianization for density estimation}
Gaussianization methods were originally developed for density estimation, where the goal is to iteratively transform data toward a Gaussian distribution. Each iteration alternates between rotations, which reduce dependence among coordinates, and marginal Gaussianization, which transforms each marginal to a standard Gaussian. In its original form~\citep{chen2000gaussianization}, rotations are determined by independent component analysis (ICA), which does not have closed-form solutions and can be slow to compute in high dimensions. An alternative is to choose the rotation matrices by PCA, but the iterative algorithm may get stuck in local optima~\citep{laparra2011iterative}. \citet{laparra2011iterative} propose to use random rotations, which makes each iteration faster to compute, but may require more iterations to converge. From the perspective of continuous-time particle flows, a related density estimator is developed by~\citep{tabak2010density}, which constructs the Gaussianization map by alternating between random rotations and simple transformations that are either coordinatewise maps or radial expansions.
Gaussianization flow~\citep{meng2020gaussianization} instead treats the
rotation matrices as trainable parameters, parametrized as products of
Householder transformations.
The flow is trained as a normalizing flow, with rotations and
componentwise transformations optimized jointly.
The authors also extend the universal approximation theorem of
\cite{chen2000gaussianization} to their specific parametrization.

\paragraph{Normalizing flows}

Normalizing flows define the variational family as the pushforward of a simple reference distribution, typically $\gamma=\N(0,I_d)$, through a diffeomorphism $\tau_\theta$ parametrized by $\theta\in\Theta$. The pushforward measure is denoted $q_\theta := \tau_\theta \# \gamma$. By the change of variable formula, the density of $q_\theta$ is given by
$q_\theta(\bfx) = \gamma(\tau_\theta^{-1}(\bfx)) \cdot |\nabla \tau_\theta^{-1}(\bfx)|$,
where $\tau_\theta^{-1}$ is the inverse of the transport map $\tau_\theta$ and $|\nabla \tau_\theta^{-1}(\bfx)|$ is the absolute value of the determinant of the Jacobian of $\tau_\theta^{-1}$ at $\bfx$. The KL divergence between $q_\theta$ and $p$ can be expressed as
\begin{align*}
    \kl{q_\theta}{p} = \EE[\gamma]{\log \gamma(\bfx ) - \log|\nabla\tau_\theta (\bfx) | -\log p(\tau_\theta(\bfx) )  }.
\end{align*}
In practice, the KL divergence is optimized via gradient-based optimization, where the gradient is estimated by Monte Carlo samples from $\gamma$.

Because evaluating the log-determinant term $\log|\nabla \tau_\theta|$ can be computationally expensive, transport maps are often parameterized so that the Jacobian is lower triangular. For example, polynomial basis expansions have been used to directly approximate the Knothe--Rosenblatt rearrangement~\citep{el2012bayesian}, but the number of parameters grows rapidly with dimension. Normalizing flows construct the map as a composition of simple transformations, each with a triangular Jacobian. Popular examples include autoregressive flows~\citep{kingma2016improved}, RealNVP~\citep{dinh2017density}, and neural spline flows~\citep{durkan2019neural}; see~\citep{papamakarios2021normalizing} for a survey. However, the structural constraints of these flows limit their expressiveness, often requiring many layers to capture complex distributions. Overparameterization can also make training unstable, and empirical studies show that very large Monte Carlo sample sizes are needed for reliable optimization via stochastic gradient descent~\citep{agrawal2025disentangling,blessing2024beyond}.

\paragraph{Low-dimensional transport maps}

The lazy map framework of \citet{brennan2020greedy} is particularly close in spirit to our work. It iteratively transforms a target distribution toward a Gaussian via a composition of low-dimensional transport maps, where at each iteration a subspace is identified using the certified dimension reduction (CDR) method of \citet{zahm2022certified} and a transport map is fit within that subspace. 
The low-dimensional transport maps in their framework are parametrized by triangular maps or autoregressive flows. 
In contrast, our method constructs coordinatewise maps, which yields MFVI subproblems that are much easier to solve and a final transport that is trivial to invert and evaluate. Moreover, the analytical tractability of MFVI lets us establish a stationarity guarantee and gradient flow structure for the iterative Gaussianization procedure, and derive a principled criterion for selecting the rotation at each step. Further comparison between the different choices of rotation is given in Section~\ref{sec-comparison-to-active-subspace} and Appendix~\ref{sec: cdr discussion}.

\section{Stationarity guarantee}\label{sec: stationarity guarantee}

A natural question regarding the validity of Algorithm~\ref{algo: iterative gaussianization} is whether it is possible to keep decreasing the KL divergence until convergence to the standard Gaussian. 
To quantify the decrease available at each step, we define the \emph{MFVI improvement}
\begin{align}\label{equ: mfvi improvement}
    \Delta_{\mfvi}(p)
    = \kl{\gamma}{p} - \inf_{F\in\calF}\;\kl{F\#\gamma}{p},
\end{align}
which is the reduction in KL divergence achieved by MFVI relative to the baseline $\kl{\gamma}{p}$. 
Since the identity map lies in $\calF$, it is always nonnegative.
At the $k$-th iteration of Algorithm~\ref{algo: iterative gaussianization}, the KL divergence decreases from $\kl{\gamma}{p^{(k-1)}}$ to $\kl{\gamma}{p^{(k)}}=\inf_{F\in\calF}\,\kl{F\#\gamma}{R_k\# p^{(k-1)}}$, so the decrease in KL divergence at iteration $k$ is exactly
\begin{align*}
\Delta_{\mfvi}(R_k\# p^{(k-1)}) .
\end{align*}

The next result shows that the standard Gaussian is the only target for which the MFVI improvement vanishes under almost every rotation.
\begin{thm}[Characterization of the stationary point]\label{thm: stationary point}
    Suppose that $p$ satisfies the standing smoothness and polynomial-growth
    assumptions in Section \ref{sec: mfvi background}.
    If $\Delta_{\mfvi}(R\# p)=0$ for almost every orthogonal matrix $R$
    with respect to the Haar measure on the orthogonal group, then $p$
    is the standard Gaussian distribution.
\end{thm}
The proof is provided in Appendix~\ref{prf: stationary point}.

Theorem~\ref{thm: stationary point} implies that whenever the current
target $p^{(k-1)}$ is not the standard Gaussian, there must exist some
rotation $R_k$ for which MFVI can strictly decrease the KL divergence.
Thus, the standard Gaussian is the only fixed point of
Algorithm~\ref{algo: iterative gaussianization} at which no rotation can
produce further improvement.

Algorithmically, the rotation $R_k$ can be chosen in several ways:
randomly, greedily to maximize
$\Delta_{\mfvi}(R\# p^{(k-1)})$, or by the PCA method
introduced later in Section~\ref{sec: pca}.
For random rotations, we define the expected MFVI improvement by
\begin{align*}
    \widebar \Delta_{\mfvi}(p) = \EE[R]{\Delta_{\mfvi}(R\# p)},
\end{align*}
where the expectation is taken over the randomness in $R$.
If $R$ is sampled from a distribution 
on the
orthogonal group
with a density that is strictly positive Haar-almost everywhere, then $\widebar \Delta_{\mfvi}(p)$ is divergence-like:
it is nonnegative and vanishes only when $p=\gamma$.
Indeed, if $\widebar \Delta_{\mfvi}(p)=0$, then
$\Delta_{\mfvi}(R\# p)=0$ for almost every $R$, and
Theorem~\ref{thm: stationary point} implies that $p=\gamma$.
The same conclusion holds for the greedy choice, with
$\widebar \Delta_{\mfvi}(p)$ defined as
$\sup_R \Delta_{\mfvi}(R\# p)$.
This argument is formalized in the following result.
\begin{prop}[First-order stationarity guarantee]
    \label{prop: stationarity guarantee}
    If the rotation $R_k$ is either chosen randomly from a distribution on the
orthogonal group
with a density that is strictly positive Haar-almost everywhere, or greedily to maximize $\Delta_{\mfvi}(R_k\# p^{(k-1)})$ at each iteration, then 
    \begin{align*}
        \widebar\Delta_{\mfvi}(p^{(k-1)}) = 0 \quad \text{if and only if}\quad p^{(k-1)}=\gamma.
    \end{align*}
    Moreover, Algorithm~\ref{algo: iterative gaussianization} guarantees that
    \begin{align*}
    \widebar\Delta_{\mfvi}(p^{(K)} ) \stackrel{a.s.}{\to} 0,\quad \EE{\min_{1\leq k\leq K}\widebar \Delta_{\mfvi}(p^{(k-1)})} = O(\frac{1}{K}),\quad \text{as } K\to\infty,
    \end{align*}
    where the almost sure convergence and expectation are over the randomness in the rotations.
\end{prop}
The proof is provided in Appendix~\ref{prf: stationarity guarantee}.

This result shows that the divergence-like quantity $\widebar \Delta_{\mfvi}(p^{(k)})$ converges to zero as the number of iterations increases, which can be understood as a first-order stationarity guarantee of the iterative algorithm~\citep{balasubramanian2022towards}. 

Although jointly optimizing over rotations and coordinatewise maps is a natural approach and has been studied in the literature~\citep{sheng2025mode}, it is computationally challenging. Each evaluation of $\Delta_{\mfvi}(R\# p)$ requires solving an inner MFVI problem, and the resulting objective over the orthogonal group is itself nonconvex. To avoid this computational burden while still obtaining an informative rotation, we now turn to an infinitesimal analysis of the KL decrease induced by coordinatewise transport, which will in turn motivate an efficient PCA-based approach for selecting rotations in Section~\ref{sec: pca}.

\section{Projected Fisher information as a local measure of MFVI improvement}\label{sec: instantaneous}

The MFVI improvement \(\Delta_{\mfvi}(R\#p)\) is the ideal criterion for choosing
the rotation, but evaluating it requires solving a full MFVI problem for each candidate
\(R\). We therefore consider a local version of the same variational problem. Starting
from the reference distribution \(\gamma\), we ask how much the KL divergence can be
decreased by an infinitesimal coordinatewise transport. The resulting quantity is the
projected Fisher information. 

\subsection{Projected Fisher information}
\label{sec: projected FI}

Given a coordinatewise vector field $\bfv:\R^d\to\R^d$, with $v_i(\bfx)=v_i(x_i)$, consider the perturbed distribution $q_\ep=(\mathrm{id}+\ep \bfv)\# \gamma$ for $\ep>0$. The instantaneous rate of change of the KL divergence at $\ep=0$ is given by
\begin{align}
\label{equ: kl instantaneous rate}
    \frac{\rd}{\rd\ep}\Big|_{\ep=0} \kl{q_\ep}{p}
    = -\EE[\gamma]{ \big\langle \bfv(\bfx), \nabla\log (p/\gamma)(\bfx) \big\rangle }.
\end{align}
Since the perturbation is applied to $\gamma$, we measure the size of $\bfv$ using the $L^2(\gamma)$-norm. The steepest descent direction is therefore the $L^2(\gamma)$-projection of
$\nabla\log(p/\gamma)$ onto the space of coordinatewise vector fields
\begin{align}\label{equ: optimal vector field}
    \bfv^*(\bfx)= 
    \Big[\EE[\gamma]{\nabla_{1}\log (p/\gamma)(\bfx) \mid x_1}, \ldots, \EE[\gamma]{\nabla_{d}\log (p/\gamma)(\bfx) \mid x_d} \Bigr],
\end{align}
and the resulting instantaneous rate of KL is given by
\begin{align*}
    \frac{\rd}{\rd\ep}\Big|_{\ep=0} \kl{q_\ep}{p}
    =-\sum_{i=1}^d \EE[\gamma]{\Bigl(\EE[\gamma]{\nabla_i\log (p/\gamma)(\bfx) \mid x_i}\Bigr)^2}=:-\tI(\gamma,p).
\end{align*}
See Appendix~\ref{sec: instantaneous analysis} for the derivation of Equations~\eqref{equ: kl instantaneous rate} and~\eqref{equ: optimal vector field}.
Without the constraint that $\bfv$ be coordinatewise, the steepest descent direction would be the full relative score $\nabla\log(p/\gamma)$, and the corresponding instantaneous rate of KL decrease would be the usual Fisher divergence $\EE[\gamma]{\|\nabla\log(p/\gamma)\|^2}$.
Thus, the quantity $\tI(\gamma,p)$ can be understood as the \emph{projected Fisher information}, which measures the instantaneous KL decrease when the velocity field is restricted to be coordinatewise. Its relationship with MFVI is studied further in~\citep{lacker2023independent}.

When performing MFVI in a rotated coordinate system, we consider vector fields of the form $\bfv(\bfx)=R\tran \bfu(R\bfx)$, where $\bfu:\R^d\to\R^d$ is coordinatewise and $R$ is a rotation matrix. For the perturbed distribution $\tq_\ep=(\mathrm{id}+\ep R\tran \bfu(R \cdot) ) \# \gamma$, we have
\[
\kl{\tq_\ep}{p} = \kl{(\mathrm{id}+\ep \bfu)\# \gamma}{p_R}
\]
by the invariance of KL divergence under invertible transformations and the rotational invariance of $\gamma$.
Therefore, the coordinatewise vector field $\bfu$ that yields the steepest descent of KL has the same form as in~\eqref{equ: optimal vector field}, with $p$ replaced by $p_R$, namely,
\begin{align}\label{equ: u*}
\bfu^*(\bfy)=\Big[\EE[\gamma]{\partial_{y_1} [ \log (p/\gamma)(R\tran \bfy)] \mid y_1},\ldots,\EE[\gamma]{\partial_{y_d}[\log (p/\gamma)(R\tran \bfy)] \mid y_d}\Bigr].
\end{align}
The optimal vector field in the original coordinates is then $\bfv^*(\bfx)=R\tran \bfu^*(R\bfx)$, and the instantaneous rate of KL decrease is $\tI(\gamma,p_R)$.

\subsection{Connection to MFVI improvement}
\label{sec: connection to mfvi improvement}

The instantaneous KL decrease is directly related to the MFVI improvement, since
\[
\Delta_{\mfvi}(p_R)
\ge \kl{\gamma}{p_R} - \kl{(\mathrm{id}+\ep \bfu^*) \# \gamma}{p_R}=
\ep \tI(\gamma,p_R)+o(\ep),\quad \ep \downarrow 0.
\]
Thus, maximizing \(\tI(\gamma,p_R)\) maximizes the first-order decrease in KL divergence available to MFVI in the rotated coordinate system.

Moreover, when the target distribution is strongly log-concave, the projected Fisher information also provides an upper bound on the MFVI improvement.
\begin{prop}[{\cite[Theorem 2.5]{lacker2023independent}}]\label{thm: pFI LSI}
    If the target distribution $p$ is $\lambda$-strongly log-concave, 
    \begin{align*}
        \Delta_{\mfvi}(p) \leq \frac{1}{2\lambda} \tI(\gamma,p).
    \end{align*}
\end{prop}
In particular, if the projected FI is small, then no
large MFVI improvement is possible. Thus \(\tI(\gamma,p)\) is a first-order
stationarity diagnostic for the MFVI objective.

For Gaussian targets, the relationship is stronger: projected FI and MFVI
improvement are equivalent up to constants. 
\begin{prop}\label{prop: pFI gaussian}
    If $p=\N(0,\Omega^{-1})$ and \(\Omega_{ii}\in[C_1,C_2]\) for some constants \(0<C_1\le C_2\), then
    \begin{align*}
        \frac{1}{4\max(C_2, 1)^2} \tI(\gamma,p)\leq \Delta_{\mfvi}(p) \leq \frac{1}{4\min(C_1,1)^2} \tI(\gamma, p).
    \end{align*}
\end{prop}
The proof is given in Appendix~\ref{prf: pFI gaussian}.

When \(p\) is close to \(\gamma\), the diagonal entries \(\Omega_{ii}\) are close to one, and the bounds above imply that
$\Delta_{\mfvi}(p) \approx\frac14\tI(\gamma,p)$.
Thus, for targets that are already partly Gaussianized, the projected FI is not only a
local descent criterion, but also closely tracks the MFVI improvement.
This is particularly relevant for iterative Gaussianization, since the transformed targets are intended to become
progressively closer to Gaussian.

\subsection{Gradient flow interpretation}
\label{sec: gradient flow}

The infinitesimal analysis above perturbs $\gamma$ along a coordinatewise vector field $\bfv$, decreasing $\kl{q_\ep}{p}$ and thus moving toward $p$. By the invariance of KL divergence under invertible transformations, this is equivalent to applying the inverse infinitesimal transport to $p$, moving the target toward $\gamma$. Iterating this step yields a continuous-time analogue of iterative Gaussianization (Algorithm~\ref{algo: iterative gaussianization}): rather than solving a full MFVI problem at each iteration, we apply at each instant an infinitesimal coordinatewise update in a rotated coordinate system.

Suppose at time $t$ the target is $p_t$, the selected rotation is $R_{p_t}$ (which may depend on $p_t$), and the rotated target is $\tp_t:=R_{p_t}\# p_t$. By~\eqref{equ: u*}, the coordinatewise perturbation of $\gamma$ with the steepest instantaneous descent of KL is
\[
\bfu_t(\bfy)=
\Big[
\EE[\gamma]{\partial_{y_1}\log(\tilde p_t/\gamma)(\bfy)\mid y_1},
\ldots,
\EE[\gamma]{\partial_{y_d}\log(\tilde p_t/\gamma)(\bfy)\mid y_d}
\Big].
\]
We then apply the inverse of the infinitesimal map $\mathrm{id}+\ep u_t$ to Gaussianize the rotated target.
Returning to the original coordinates, the updated target is given by
\[
    p_{t+\ep}=
    (\mathrm{id}+\ep\bfv_t)^{-1}\#p_t,
    \qquad \bfv_t(\bfx) = R_{p_t}\tran \bfu_t(R_{p_t}\bfx).
\]
Taking the limit $\ep\downarrow0$ formally yields the continuity equation
\begin{align}\label{equ: continuity equation}
    \partial_t p_t
    =\nabla\cdot(p_t\bfv_t),
    \qquad \bfv_t(\bfx) = R_{p_t}\tran \bfu_t(R_{p_t}\bfx).
\end{align}
See Appendix~\ref{sec: instantaneous analysis}.
This flow is thus the infinitesimal-step analogue of iterative
Gaussianization, obtained by replacing each finite MFVI update with the
steepest coordinatewise infinitesimal update in the selected rotated
coordinate system.

The dynamics~\eqref{equ: continuity equation} can be interpreted as a constrained and weighted Wasserstein gradient flow of the functional
$\calE(p)=\kl{\gamma}{p}$,
where the admissible velocity fields are restricted to rotated coordinatewise directions.
With a suitable weighted metric on this constrained tangent space, the
gradient flow has the energy dissipation
\[
    \frac{\rd}{\rd t}\calE(p_t)
    =
    -\tI(\gamma,R_{p_t}\#p_t).
\]
The formal construction of the constrained and weighted Wasserstein gradient flow structure is given in Appendix~\ref{sec: gradient flow derivation}. 
If $\tI(\gamma, R_{p_t}\# p_t)$ is uniformly continuous in $t$, then $\tI(\gamma, R_{p_t}\# p_t)$ converges to zero along this flow.
If the rotation is chosen greedily, or randomly with a density that is strictly positive Haar-almost everywhere, then the corresponding
maximum or averaged projected FI is zero if and only if \(p_t=\gamma\). This identifies
\(\gamma\) as the only possible stationary point of the gradient flow, 
providing a continuous-time analogue of the stationarity guarantee in
Proposition~\ref{prop: stationarity guarantee}.
Moreover, if a Polyak--{\L}ojasiewicz-type
inequality holds, i.e. $\tI(\gamma,R_p\#p)\geq c\,\kl{\gamma}{p}$ for some $c>0$, then Gronwall's inequality gives exponential convergence of $\kl{\gamma}{p_t}$ to zero. Conditions under which such a
PL-type inequality may hold
are discussed in Appendix~\ref{LSI discussion}.

This gradient flow structure further clarifies the role of the projected FI. It is the steepest KL decrease rate available under the same structural restriction as MFVI, namely transport by coordinatewise maps after a rotation.
Consequently, choosing rotations that make $\tI(\gamma,R\#p)$ large amounts to choosing
coordinate systems in which the target can be Gaussianized most rapidly at the
infinitesimal scale.

\section{Rotation selection via relative score PCA}
\label{sec: pca}

The previous section identifies projected Fisher information as the
local KL decrease rate available to MFVI in a rotated coordinate system.
This suggests choosing a rotation $R$ that makes
$\tI(\gamma,R\#p)$ large.
We now derive an efficient selection rule by exploiting a structural
decomposition of $\tI(\gamma,R\#p)$.

\subsection{Decomposition of projected Fisher information}
\label{sec: FI decomposition}

The projected FI is the squared $L^2(\gamma)$-norm of the relative score $h(\bfx)=\nabla\log (p/\gamma)(\bfx)$ after projection onto coordinatewise functions. 
The multivariate Hermite polynomials form a complete orthonormal basis of $L^2(\gamma)$, and expanding the relative score in this basis yields a degree-wise decomposition of the projected FI, as stated in the theorem below.

We use $\bfe_j$ to denote the $j$-th standard basis vector in $\R^d$, and $\bfk! = \prod_{j=1}^d k_j!$ for a multi-index $\bfk\in\bbN^d$. The multivariate Hermite polynomials $\{\he_{\bfk}\}_{\bfk\in\bbN^d}$ are defined as the product of univariate Hermite polynomials, i.e. $\he_{\bfk}(\bfx) = \prod_{j=1}^d \he_{k_j}(x_j)$, where $\he_{k_j}$ is the normalized probabilist's Hermite polynomial of degree $k_j$.
The proof of the following theorem is given in Appendix~\ref{prf: pFI decomposition}.
\begin{thm}\label{thm: pFI decomposition}
    Let $p$ be a density supported on $\R^d$ with $r=\log (p/\gamma)\in L^2(\gamma),h=\nabla\log(p/\gamma)\in L^2(\gamma)$.
    For an orthogonal matrix $R\in\R^{d\times d}$ with rows $R_i\tran$, $1\leq i\leq d$, we have
    \begin{align*}
        \tI(\gamma, R\# p) = I^{(1)} + \sum_{m\geq 2} I^{(m)}(R),\quad \text{where}\quad I^{(m)}(R) = \sum_{i=1}^d \big\langle \calA^{(m)}, R_i^{\otimes m} \big\rangle^2,
    \end{align*}
    with $I^{(1)}=\|\EE[\gamma]{h(\bfx)}\|_2^2$ not depending on $R$, and $\calA^{(m)}$ being an order-$m$ tensor defined as
    \begin{align*}
    \calA^{(m)}_{j_1,\ldots,j_m} = \frac{\sqrt{\bfk!}}{\sqrt{(m-1)!}} \EE[\gamma]{r(\bfx) \he_{\bfk}(\bfx)},\;\;\text{where}\;\; \bfk=\sum_{\ell=1}^m \bfe_{j_\ell} \in \bbN^d.
    \end{align*}
    If additionally all the partial derivatives of $r$ up to order $m$ are in $L^2(\gamma)$, then $\calA^{(m)}$ can be equivalently expressed as
    \begin{align*}
    \calA^{(m)}_{j_1,\ldots,j_m} = \frac{1}{\sqrt{(m-1)!}} \EE[\gamma]{\partial_{j_1}\cdots \partial_{j_m} r(\bfx)}.
    \end{align*}
\end{thm}
Theorem~\ref{thm: pFI decomposition} gives an ANOVA-type decomposition
of $\tI(\gamma,R\#p)$ into orthogonal degree-wise contributions.
The term $I^{(1)}$ is invariant under rotations, while
$I^{(m)}(R)$, for $m\geq 2$, is the contribution associated with the
degree-$(m-1)$ Hermite component of the relative score, equivalently
the degree-$m$ Hermite component of the log-density ratio.

\subsection{Relative score PCA}
The decomposition in Theorem~\ref{thm: pFI decomposition} suggests
first targeting the lowest-degree term $I^{(2)}(R)$.
This term has a particularly simple form and maximizing it reduces to a
matrix eigenvalue problem.

\begin{prop}\label{prop: pca}
    We have
    \begin{align*}
        I^{(2)}(R) = \sum_{i=1}^d (R H R\tran)_{ii}^2,\quad \text{where } H=\EE[\gamma]{\bfx \, \nabla\log(p/\gamma)(\bfx)\tran}.
    \end{align*}
    Moreover, $I^{(2)}(R)$ is maximized when $RHR\tran$ is diagonalized.
\end{prop}

Proposition~\ref{prop: pca} leads to a simple rotation-selection rule:
choose the rows of $R$ to be the eigenvectors of $H$.
Although $I^{(2)}(R)$ captures only the linear component of the relative
score, this contribution is often substantial in practice and yields a
computationally efficient choice of rotation.

Moreover, for rotated product distributions covered by the following
exact-recovery result, the eigenvectors of $H$
recover the product coordinate system exactly; in that case, the
subsequent mean-field approximation is exact. This is stated in the following result, whose proof is given in Appendix~\ref{prf: rotated product exact recovery}.

\begin{prop}[Exact recovery for rotated product distributions]\label{prop: rotated product exact recovery}
Suppose $p=Q\# \tp$, where $Q\in O(d)$ and $\tp(\bfx)=\prod_{i=1}^d \tp_i(x_i)$ is a product distribution. 
If $\EE[\gamma]{x \partial_x \log \tp_i(x) }$ are distinct for $1\leq i\leq d$, 
then the eigenvectors of $H$ coincide with the columns of $Q$, up to permutations and sign flips.
\end{prop}

In practice, the matrix $H$ can be estimated via Monte Carlo sampling:
\begin{align*}
    \widehat H = \frac{1}{N} \sum_{i=1}^N \bfx_i \Big(\nabla \log p(\bfx_i) + \bfx_i \Big) \tran,\quad \bfx_i\iid \gamma,\quad 1\leq i\leq N,
\end{align*}
and then symmetrized as $(\widehat H + \widehat H\tran)/2$. An alternative unbiased estimator is \sloppy{$\tfrac1N\sum_{i=1}^N \bfx_i \nabla \log p(\bfx_i)\tran + I_d$}, but the form given above may yield lower variance when $p$ is close to $\gamma$, since in that regime $\nabla \log p(\bfx) + \bfx$ is close to zero.

Rather than computing the full eigendecomposition of $\widehat H$, it is often sufficient to compute only the leading $r$ eigenvectors associated with the largest eigenvalues in magnitude. These eigenvectors then form the first $r$ columns of $R\tran$. We can encode such an orthogonal matrix $R$ efficiently as a product of $r$ Householder reflections $I_d-2w_j w_j\tran$ for unit vectors $w_j$, $1\leq j\leq r$. This representation requires only $O(rd)$ storage and $O(r d)$ computation to apply $R$ or $R\tran$ to a vector in $\R^d$. 

Because $H$ is the cross-covariance matrix between $\bfx\sim\gamma$ and the relative score $\nabla\log p(\bfx)+\bfx$, we term this approach \emph{relative score PCA}. A summary of the procedure is given in Algorithm~\ref{algo: pca}.

\begin{algorithm}
\caption{Relative score PCA}
\label{algo: pca}
\begin{algorithmic}
    \REQUIRE{Target distribution $p$; number of PCs $r$; Monte Carlo sample size $N$}
    \STATE{Generate $\bfx_i\iid \N(0,I_d)$ for $1\leq i\leq N$}
    \vspace{.5em}
    \STATE{Compute $\widehat{H}=\frac1N\sum_{i=1}^N \bfx_i (\nabla\log p(\bfx_i) + \bfx_i)\tran $ and symmetrize as $\widehat H= (\widehat H+\widehat H\tran )/2$}
    \vspace{.5em}
    \RETURN{Top $r$ eigenvectors of $\widehat H$ corresponding to the largest eigenvalues in magnitude if $\widehat{H}\neq 0$, otherwise use a prescribed fallback such as a Haar-random rotation.}
\end{algorithmic}
\end{algorithm}

\begin{rmk}[Higher-degree contribution]\label{rmk: higher degree contribution}
When the projected relative score is dominated by nonlinear components,
the matrix $H$ may not be informative.
In this case, one option is to fall back to a random rotation, which
still yields a strict KL decrease in expectation.
Another option is to target a higher-degree contribution
$I^{(m)}(R)=\sum_{i=1}^d \langle \calA^{(m)}, R_i^{\otimes m} \rangle^2$.
For $m\geq 3$, maximizing $I^{(m)}(R)$ leads to a tensor PCA problem,
which is NP-hard in general~\citep{hillar2013most}.
Nevertheless, practical heuristics are available, such as unfolding the
tensor into a matrix and computing its leading left singular vectors.
We show in Appendix~\ref{sec: tensor pca} that this unfolding
heuristic can be implemented without explicitly forming the tensor:
it amounts to probing the relative score $h$ with degree-$(m-1)$
Hermite polynomials along random directions, and then computing the
principal components of the resulting vectors.
\end{rmk}

\subsection{Connection to Stein discrepancy}
\label{sec: stein}

The degree-$m$ term $I^{(m)}(R)$ also has an interpretation as a
Stein discrepancy over a restricted class of coordinatewise test
functions.
For a vector-valued test function $\bfv:\R^d\to\R^d$ and a distribution $p$, the Stein operator is defined as
\begin{align*}
    \calT_p \bfv(\bfx) = \nabla \log p(\bfx)\tran \bfv(\bfx) + \nabla\cdot \bfv(\bfx). 
\end{align*}    
For sufficiently regular $\bfv$, Stein's identity gives $\EE[p]{\calT_p \bfv(\bfx)}=0$~\citep{gorham2015measuring}. The squared Stein discrepancy between $\gamma$ and $p$ with respect to a given test function $\bfv$ is $\calS_p^\gamma(\bfv)^2$, where $\calS_p^\gamma(\bfv):=\EE[\gamma]{\calT_p \bfv(\bfx)}$.
It measures how much Stein's identity is violated under $\gamma$, with larger values indicating a greater discrepancy between the two distributions. 

This quantity is directly related to the infinitesimal KL decrease considered earlier. 
Indeed, integration by parts gives $\calS_p^\gamma(\bfv) = \EE[\gamma]{\nabla \log (p/\gamma)(\bfx)\tran \bfv(\bfx)} $.
Therefore, Equation~\eqref{equ: kl instantaneous rate} can be equivalently expressed as
\begin{align*}
    \frac{\rd}{\rd\ep}\Big|_{\ep=0} \kl{(\mathrm{id} + \ep \bfv)\# \gamma}{p}
    = -\calS_p^\gamma(\bfv) .
\end{align*}
Thus, a test function with a large positive Stein discrepancy determines a perturbation direction along which the KL divergence decreases rapidly. 
This observation is also the variational principle underlying Stein variational inference algorithms~\citep{liu2016stein}.

We now consider the class of coordinatewise test functions given by $\Psi_{m-1}=\big\{\bfu_{\bfa}(\bfx)=\sum_{i=1}^d a_i \he_{m-1}(x_i) \bfe_i:\; \bfa\in\R^d, \|\bfa\|\leq 1\big\}$. 
This class probes the degree-$(m-1)$ component of the relative score
along each coordinate.
For $\bfu_\bfa\in \Psi_{m-1}$, we have 
\begin{align}\label{equ: stein discrepancy expression}
    \calS^\gamma_{R\# p}(\bfu_{\bfa}) =\sum_{i=1}^d a_i \langle \calA^{(m)}, R_i^{\otimes m} \rangle.
\end{align}
See Appendix~\ref{prf: stein discrepancy expression} for the derivation.
Taking the supremum over $\|\bfa\|\leq 1$ yields
\begin{align*}
    \sup_{\bfu_{\bfa} \in \Psi_{m-1} }\calS^\gamma_{R\# p}(\bfu_{\bfa}) = \Big(\sum_{i=1}^d \langle \calA^{(m)}, R_i^{\otimes m} \rangle^2\Big)^{1/2} = \sqrt{I^{(m)}(R)}.
\end{align*}
In other words, $I^{(m)}(R)$ is the squared Stein discrepancy between $\gamma$ and $R\# p$ over the class of coordinatewise Hermite polynomial test functions of degree $m-1$. 
Therefore, maximizing $I^{(m)}(R)$ amounts to choosing the rotation that maximizes this Stein discrepancy, or equivalently, the instantaneous KL decrease along these perturbations.

The linear case gives a complementary interpretation of the matrix $H=\EE[\gamma]{ \bfx h(\bfx) \tran}$ used in relative score PCA.
For the linear test function $\bfv(\bfx)=L\bfx$ for some $L\in\R^{d\times d}$, we have $\calS^\gamma_{p}(\bfv) = \tr(HL)$.
Thus, the matrix $H$ fully encodes the Stein discrepancy between $\gamma$ and $p$ over linear test functions. In particular, if $L=\theta\theta\tran$ for a unit vector $\theta$, then $\calS^\gamma_p(\bfv) = \theta\tran H \theta$.
This shows that the eigenvectors of $H$ identify the principal directions along which $p$ deviates the most from $\gamma$ as measured by linear Stein discrepancy. 
The sign and magnitude of the associated eigenvalues indicate both the nature (concentration vs.~flatness) and the severity of the deviation. Therefore, aligning the coordinate system with the eigenvectors of $H$ directly targets the most non-Gaussian directions in the sense of linear Stein discrepancy.

\subsection{Other gradient-based PCA methods}
\label{sec-comparison-to-active-subspace}
Using PCA-based techniques to identify important directions for a function or probability distribution is a well-established strategy. 
However, existing methods typically focus on identifying low-dimensional subspaces for function approximation or likelihood-informed posterior approximation, whereas the proposed method aims to identify the rotation that makes coordinatewise MFVI more effective.

For instance, in constructing low-dimensional response surfaces for high-dimensional functions $f$, the active subspace method~\citep{constantine2014active} identifies directions of greatest variability by computing the principal components of $\EE{\nabla_\bfx f(\bfx) \nabla_{\bfx} f(\bfx)\tran }$.
In Bayesian analysis, PCA-based dimension reduction has also been widely explored~\citep{cui2014likelihood,zahm2022certified}. The certified dimension reduction (CDR) method~\citep{zahm2022certified} seeks to approximate a distribution $\nu$ by modifying another distribution $\mu$ along a low-dimensional subspace. Specifically, CDR considers approximations of the form
$\nu_{R_r,\ell}(\bfx)\propto\ell(R_r\tran \bfx)\cdot \mu(\bfx)$,
where $R_r\in\R^{d\times r}$ is an orthogonal projection onto an $r$-dimensional subspace and $\ell$ is a profile function. In the Bayesian setting, $\nu$ is the posterior, $\mu$ is the prior, and such approximations are effective when the likelihood updates the prior primarily along a low-dimensional subspace.

When $\mu$ satisfies a subspace logarithmic Sobolev inequality, this inequality provides an upper bound on the reconstruction error $\inf_{\ell} \kl{\nu}{\nu_{R_r,\ell}}$. The bound is minimized when the columns of $R_r$ are chosen as the leading eigenvectors of the relative Fisher information matrix
\begin{align*}
    \FI(\nu,\mu)=\EE[\nu]{\nabla\log\frac{\nu}{\mu} (\nabla\log \frac{\nu}{\mu})\tran },
\end{align*}
which is the central idea behind CDR~\citep{zahm2022certified}.

In our setting, if we take $\nu=p$ as the target and $\mu=\gamma$ as the standard Gaussian, the matrix $\FI(p,\gamma)$ is not tractable, since evaluating it requires samples from $p$---the very challenge we aim to overcome. This is also why CDR relies on iterative procedures. A natural alternative is to swap the roles of $p$ and $\gamma$, considering instead $\FI(\gamma,p)$, which can be estimated by sampling from $\gamma$. However, although the leading eigenvectors of $\FI(\gamma,p)$ can still be informative for constructing low-dimensional approximation, these directions are not necessarily well-suited for constructing mean-field approximations.

For the rotated product distribution $p=Q\# \tp$ considered in Proposition~\ref{prop: rotated product exact recovery}, the optimal rotation is $Q\tran$, and relative score PCA recovers it exactly under the distinct-eigenvalue condition in that proposition. However, the eigenvectors of $\FI(\gamma,p)$ generally do not recover $Q\tran$. The reason is that, for the product distribution $\tp$, $\FI(\gamma, \tp)$ is in general not diagonal. 
For instance, if $\tp=\N(\mu,S)$ where $S$ is diagonal, then $\FI(\gamma,\tp)=(S^{-1}-I_d)^2 + S^{-1}\mu\mu\tran S^{-1}$, which depends on the mean $\mu$ and thus is not diagonal in general. 
Although one can modify the definition of $\FI(\gamma, p)$ to remove this mean dependence, the resulting matrix can still fail to recover the correct rotation. See further discussion in Appendix~\ref{sec: cdr discussion}.

\section{Numerical results}
\label{sec: experiments}

We evaluate the effectiveness of the proposed algorithm on several posterior sampling tasks. 
The experiments are designed to evaluate two aspects of the proposed methodology.
Sections~\ref{sec: logistic regression}--\ref{sec: poisson glmm} assess the one-step benefit of relative score PCA rotations for MFVI,
while Sections~\ref{sec: sparse logistic}--\ref{sec: irt} assess the full iterative Gaussianization procedure as a
flow-like approximation method.
Code for reproducing the experiments is available at \url{https://github.com/liusf15/iterative-gaussianization.git}.

For MFVI, the componentwise transformations are parametrized by rational quadratic splines~\citep{durkan2019neural} with 10 knots, with range boundaries set to $(-8, 8)$. When selecting rotations using relative score PCA (Algorithm~\ref{algo: pca}), we use a Monte Carlo sample size of 1000, and retain the top principal components that explain 95\% of the variance unless otherwise noted. Optimization of variational objectives is carried out with Adam~\citep{kingma2014adam}, using a learning rate of 0.01 and a fixed Monte Carlo sample of size 1000. For evaluation, we generate 2000 samples from the learned variational distribution to compute metrics such as maximum mean discrepancy (MMD), kernelized Stein discrepancy (KSD), effective sample size (ESS), and evidence lower bound (ELBO). When reference samples from the target distribution are needed, we generate 2000 samples using the No-U-Turn Sampler (NUTS; \citep{hoffman2014no, carpenter2017stan}) with 20 chains, 25,000 burnin iterations and 50,000 sampling iterations with thinning, unless otherwise noted. To account for randomness, each experiment is repeated 20 times independently, with random initialization and random training/evaluation samples.

Since the performance of reverse KL-based variational inference is sensitive to initialization, we employ a Laplace approximation to obtain a good starting point. Specifically, we find the minimum $\bfx^*$ of the potential function $U(\bfx)=-\log p(\bfx)$, and the corresponding inverse Hessian matrix $C^*=(\nabla^2 U(\bfx^*))^{-1}$. We then shift the target distribution by $\bfx^*$ and rescale each coordinate by $(C^*_{i,i})^{1/2}$, thereby centering the target near the origin and approximately normalizing the marginal variance. A similar initialization strategy was used in~\citep{agrawal2025disentangling}.

\subsection{Bayesian logistic regression}
\label{sec: logistic regression}
The Bayesian logistic regression model is defined as
\begin{align*}
    &\beta\sim \N(0, \sigma^2 I_d),\quad y_i\mid x_i,\beta\sim \mathrm{Bernoulli}(S(x_i\tran\beta)),\quad 1\leq i\leq n,
\end{align*}
where $S(x)=(1+e^{-x})^{-1}$ is the sigmoid function, $x_i\in\R^d$ represents the covariates, and $y_i\in\{0,1\}$ represents the binary response. We take $n=20$, $d=10$, and fix $\sigma=2$. The covariates are generated as $x_i\iid\N(0,\Sigma_X)$, where $\Sigma_X=UDU\tran$, $U$ is a random orthogonal matrix, and $D$ is diagonal with entries equally spaced on a log scale from $10^{-1}$ to $10^1$. The responses $y_i$ are drawn i.i.d. from $\mathrm{Bernoulli}(0.5)$.

We first compare standard MFVI with rotated MFVI using the proposed relative score PCA method (Algorithm~\ref{algo: pca}). Figure~\ref{fig: logistic scatter} shows samples from the learned variational distribution in red and reference samples from the target distribution in blue. Standard MFVI, shown in the left panel, substantially underestimates the posterior variance. In contrast, MFVI with the PCA-based rotation, shown in the right panel, provides a much closer approximation to the target posterior.

\begin{figure}
    \centering
    \includegraphics[width=0.5\textwidth]{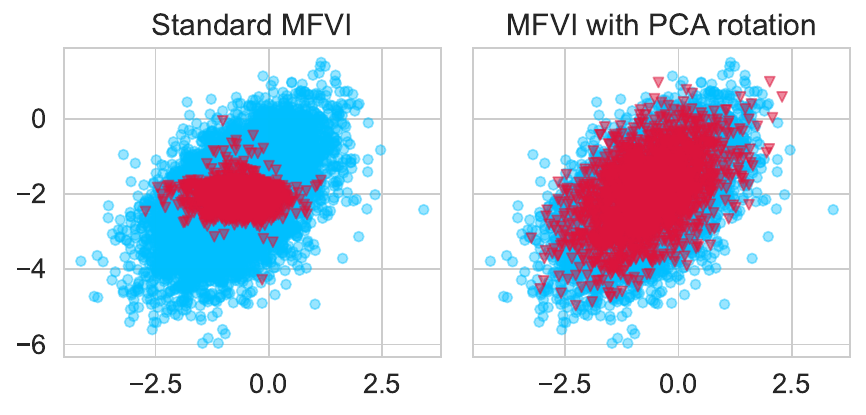}
    \caption{MFVI with different rotation choices in the Bayesian logistic regression example. The left panel shows standard MFVI without rotation, while the right panel shows MFVI with the proposed relative score PCA rotation. The red dots are samples from the learned variational distribution, and the blue dots are reference samples from the target obtained using NUTS.}
    \label{fig: logistic scatter}
\end{figure}

We also compare the proposed relative score PCA with the rotation obtained from the eigenvectors of the covariance of the relative score, as discussed in Section~\ref{sec-comparison-to-active-subspace} and Appendix~\ref{sec: cdr discussion}. We refer to this alternative as the Fisher information rotation. 
Figure~\ref{fig: logistic training} shows the MFVI training curves under different rotation choices. The proposed relative score PCA consistently achieves a lower loss, or equivalently a higher ELBO, than the Fisher information rotation. As the sample size $n$ increases, the posterior becomes closer to Gaussian, and the performance gap between the two rotations is expected to narrow.

\begin{figure}
    \centering
    \includegraphics[width=.8\textwidth]{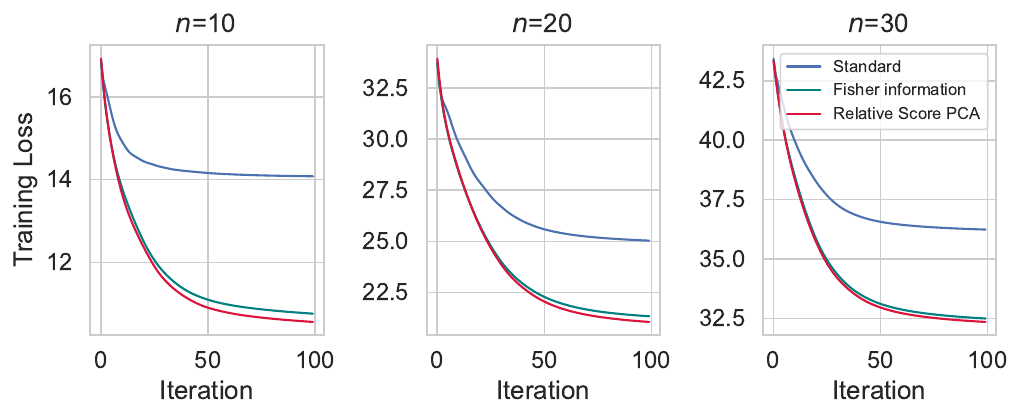}
    \caption{Training curves for MFVI with different rotation choices in the Bayesian logistic regression example. From left to right, the sample sizes are $n=10,20,30$.}
    \label{fig: logistic training}
\end{figure}

\subsection{Posteriordb benchmarks}
\label{sec: posteriordb}

We further evaluate our method on posterior distributions from \texttt{posteriordb}~\citep{posteriordb}, a repository of benchmark Bayesian models. We consider 12 low-dimensional targets with dimensions ranging from 2 to 8, and compare standard MFVI against PCA-based MFVI (using the top PCs that explain 95\% of the variance).

For each target, we report the ELBO and MMD in Table~\ref{tab: posteriordb}. The MMD is computed with 2000 reference samples drawn from NUTS. We report the mean and standard deviation (in parentheses) of each metric over 20 replicates. Additional metrics including KSD and ESS are reported in Appendix~\ref{app: posteriordb}.

We observe that PCA-based MFVI generally outperforms standard MFVI, often by a large margin. Given the small computational overhead of PCA, performing MFVI in the PCA-based coordinate system is a simple yet powerful enhancement to standard MFVI.

\begin{table}

\centering
\begin{tabular}{lcc|cc}
\toprule
 & \multicolumn{2}{c|}{ELBO} & \multicolumn{2}{c}{MMD} \\
 & MF & PCA & MF & PCA \\
\midrule
M0 & \phantom{00}-243.8 (0.0) & \phantom{00}\textbf{-243.7} (0.0) & \phantom{0}0.065 (0.005) & \phantom{0}\textbf{0.014} (0.014) \\
arK & \phantom{0000}88.3 (0.1) & \phantom{0000}\textbf{92.3} (0.1) & \phantom{0}0.241 (0.002) & \phantom{0}\textbf{0.087} (0.004) \\
garch & \phantom{00}-448.4 (0.0) & \phantom{00}\textbf{-447.8} (0.0) & \phantom{0}0.218 (0.005) & \phantom{0}\textbf{0.146} (0.013) \\
gp-regr & \phantom{000}\textbf{-23.5} (0.0) & \phantom{000}\textbf{-23.5} (0.0) & \phantom{0}\textbf{0.010} (0.010) & \phantom{0}0.015 (0.014) \\
hmm & \phantom{00}-167.6 (0.0) & \phantom{00}\textbf{-166.8} (0.0) & \phantom{0}\textbf{0.016} (0.013) & \phantom{0}0.036 (0.008) \\
kidscore-interaction & \phantom{0}-1888.3 (1.1) & \phantom{0}\textbf{-1884.3} (1.5) & \phantom{0}0.399 (0.010) & \phantom{0}\textbf{0.032} (0.012) \\
mesquite & \phantom{00}-317.6 (0.1) & \phantom{00}\textbf{-311.2} (0.2) & \phantom{0}0.270 (0.002) & \phantom{0}\textbf{0.092} (0.005) \\
nes-logit & \phantom{00}-781.8 (0.0) & \phantom{00}\textbf{-780.7} (0.0) & \phantom{0}0.339 (0.004) & \phantom{0}\textbf{0.015} (0.016) \\
normal-mixture & \phantom{0}-2098.9 (0.0) & \phantom{0}\textbf{-2098.7} (0.0) & \phantom{0}0.045 (0.006) & \phantom{0}\textbf{0.024} (0.010) \\
radon & -18562.2 (0.0) & \textbf{-18562.1} (0.0) & \phantom{0}0.055 (0.004) & \phantom{0}\textbf{0.013} (0.010) \\
sesame & \phantom{00}-107.1 (0.0) & \phantom{00}\textbf{-106.6} (0.0) & \phantom{0}0.151 (0.004) & \phantom{0}\textbf{0.018} (0.016) \\
wells & \phantom{0}-1954.1 (0.0) & \phantom{0}\textbf{-1952.6} (0.0) & \phantom{0}0.214 (0.004) & \phantom{0}\textbf{0.039} (0.006) \\
\bottomrule
\end{tabular}
\caption{\label{tab: posteriordb} Average performance metrics for \texttt{posteriordb} experiments, with standard deviations over 20 independent replicates shown in parentheses. }
\end{table}

\subsection{Poisson generalized linear mixed model}
\label{sec: poisson glmm}

Generalized linear mixed models (GLMMs) are widely used to model grouped data, yet full Bayesian inference for such models can be computationally challenging due to the high-dimensional random effects and strong dependencies between global and local parameters~\citep{fong2010bayesian}. Appropriate model reparametrizations can markedly improve both MCMC mixing~\citep{papaspiliopoulos2007general} and variational inference accuracy~\citep{tan2021use}. We consider the Poisson GLMM studied in~\citep[Section 8.1]{tan2021use}:
\begin{align*}
    &\beta_0,\beta_1 \sim \N(0,100), \\
    &\sigma^2 \sim \mathrm{InvGamma}(0.5, 0.5), \quad b_i\sim \N(0, \sigma^2),\quad 1\leq i\leq n, \\
    &y_{ij} \sim \mathrm{Poisson}(e^{\eta_{ij}} ),\quad \eta_{ij} = \beta_0 + \beta_1 x_{ij} + b_i,\quad 1\leq j\leq m.
\end{align*}
The global parameters are denoted by $\btheta_G=(\sigma,\beta_0,\beta_1)$, and the local random effects by $b_i$ ($1\leq i\leq n$). Data are generated as in~\citep{tan2021use} with $n=500$ and $m=7$, giving a total dimension of 503. Following~\citep{tan2021use}, the local effects are reparametrized as $\tb_i=(b_i - \lambda_i)/L_i$, where $\lambda_i, L_i$ depend on $\btheta_G$ and are obtained from a Gaussian approximation of the conditional distribution of $b_i\mid \btheta_G, \bfy_i$; see Appendix~\ref{sec: glmm details} for details.

For the reparametrized posterior, \citet{tan2021use} considers a Gaussian variational approximation of the form $q_G(\btheta_G) \cdot \prod_{i=1}^n q_{\tb_i}(\tb_i) $, where $q_G$ is a full-covariance Gaussian and each $q_{\tb_i}$ is a univariate Gaussian. We apply our rotated variational inference to the same reparametrized target. To keep the number of parameters comparable and demonstrate the effectiveness of relative score PCA, we approximate the rotated target with a product Gaussian distribution. 

Figure~\ref{fig: glmm} shows the posterior distributions of the global parameters $(\sigma,\beta_0,\beta_1)$, and Table~\ref{table: glmm global} reports the posterior means and standard deviations averaged over 20 independent replicates. The two VI methods, Gaussian VI and MFVI+PCA, are compared against MCMC. For the regression coefficients $\beta_0$ and $\beta_1$, both VI methods match the MCMC posteriors reasonably well. For the variance component $\sigma$, however, Gaussian VI severely underestimates the right tail of the posterior, leading to a biased posterior mean and underestimated variance. MFVI+PCA provides a noticeably better fit to the posterior of $\sigma$, with more accurate mean and variance estimates. It also yields a smaller average MSE in estimating the posterior means and standard deviations of the local effects ${b_i}$ (Table~\ref{table: glmm local} in Appendix~\ref{sec: glmm details}). In this Poisson GLMM example, the reparametrization does not fully eliminate the dependence between the global parameters $\btheta_G$ and the local effects $\tb_i$. Unlike Gaussian VI, which ignores this dependence, the proposed method exploits score information to construct rotations that better capture the global-local dependence structure.

\begin{figure}
    \centering
    \includegraphics[width=0.8\textwidth]{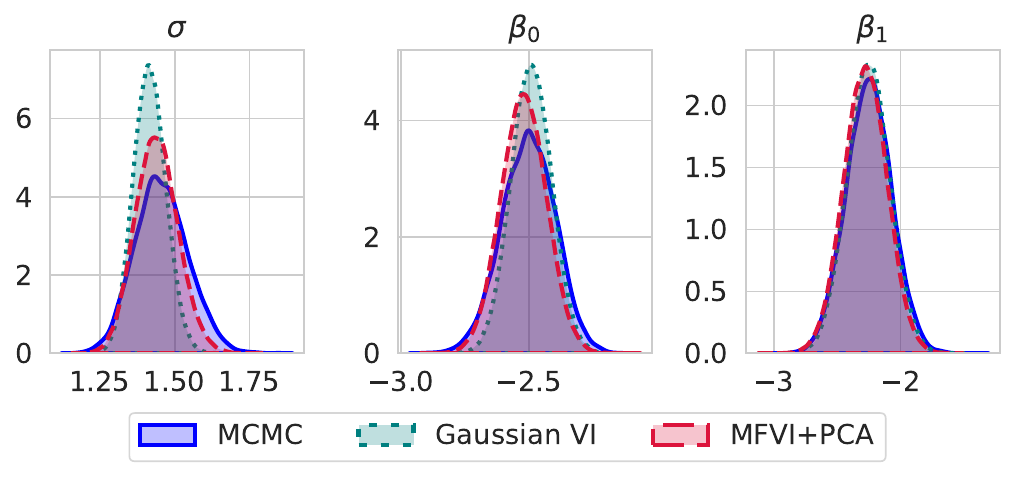}
    \caption{Posterior distributions of the global parameters in the Poisson GLMM example.}
    \label{fig: glmm}
\end{figure}

\begin{table}  
    \centering
    \begin{tabular}{lccc}
    \toprule
     & Gaussian VI & MFVI+PCA & MCMC \\
    \midrule
    $\sigma$ & \phantom{-}1.42 $\pm$ 0.05 & \phantom{-}1.44 $\pm$ 0.07 & \phantom{-}1.46 $\pm$ 0.09 \\
    $\beta_0$ & -2.49 $\pm$ 0.08 & -2.52 $\pm$ 0.09 & -2.50 $\pm$ 0.10 \\
    $\beta_1$ & -2.25 $\pm$ 0.17 & -2.28 $\pm$ 0.17 & -2.25 $\pm$ 0.18 \\
    \bottomrule
\end{tabular}
\caption{\label{table: glmm global} Posterior means and posterior standard deviations for the global parameters in Poisson GLMM.}
\end{table}

\subsection{Sparse logistic regression}
\label{sec: sparse logistic}
We consider the sparse logistic regression model given by
\begin{align*}
    \tau & \sim \mathrm{Gamma}(0.5,\, 0.5),\\
    \lambda_j & \iid \mathrm{Gamma}(0.5,\, 0.5),\quad \beta_j \iid \N(0,1),\quad 1\leq j\leq d,\\
    y_i & \sim \mathrm{Bernoulli}(S(\tau\cdot \bfx_i\tran (\bbeta \circ \blambda ) ) ),\quad 1\leq i\leq n.
\end{align*}
Here, $\tau\in\R_+$ is the global scale parameter, $\blambda\in\R_+^d$ is the per-dimension scale, and $\bbeta\in\R^d$ is the regression coefficient. We fit the model on the German credit dataset~\citep{german_credit}, which has 25 numeric covariates including an intercept. The dimension of the posterior distribution of $(\tau,\blambda,\bbeta)$ is 51. 

We compare the proposed method with the neural spline flow (NSF)~\citep{durkan2019neural}, a popular normalizing flow model. To ensure fair comparison, we match the number of layers in NSF with the number of iterations in our method, and use the same number of spline knots in both. Each neural network in NSF has 1 hidden layer of size 5, resulting in about five times as many parameters as our method. Both methods are trained with 1000 samples using Adam with learning rate 0.01. NSF is optimized for 200 iterations, while our method uses 100 iterations per layer, keeping the wall-clock times on the same order of magnitude.

The results are shown in Figure~\ref{fig: sparse logistic}, with error bars indicating variation across 20 independent runs. We vary the number of layers/iterations over ${2,4,6}$ on the $x$-axis. We observe that the proposed method significantly outperforms NSF in terms of KSD, MMD, and MSE for estimating first and second moments, while requiring less computation time. The gap is especially pronounced when the number of layers is small.

One could potentially improve the performance of NSF by increasing the number of layers or enlarging the capacity of the neural networks. However, training a more complex model comes with higher computational cost and requires careful hyperparameter tuning. Furthermore, any change in the architecture requires retraining the entire model. In contrast, the proposed method consists only of solving multiple MFVI problems, which are much easier to optimize and require minimal tuning. Moreover, if the number of iterations is deemed insufficient, additional iterations can be added seamlessly without retraining the earlier ones.

\begin{figure}
    \centering
    \includegraphics[width=.8\textwidth]{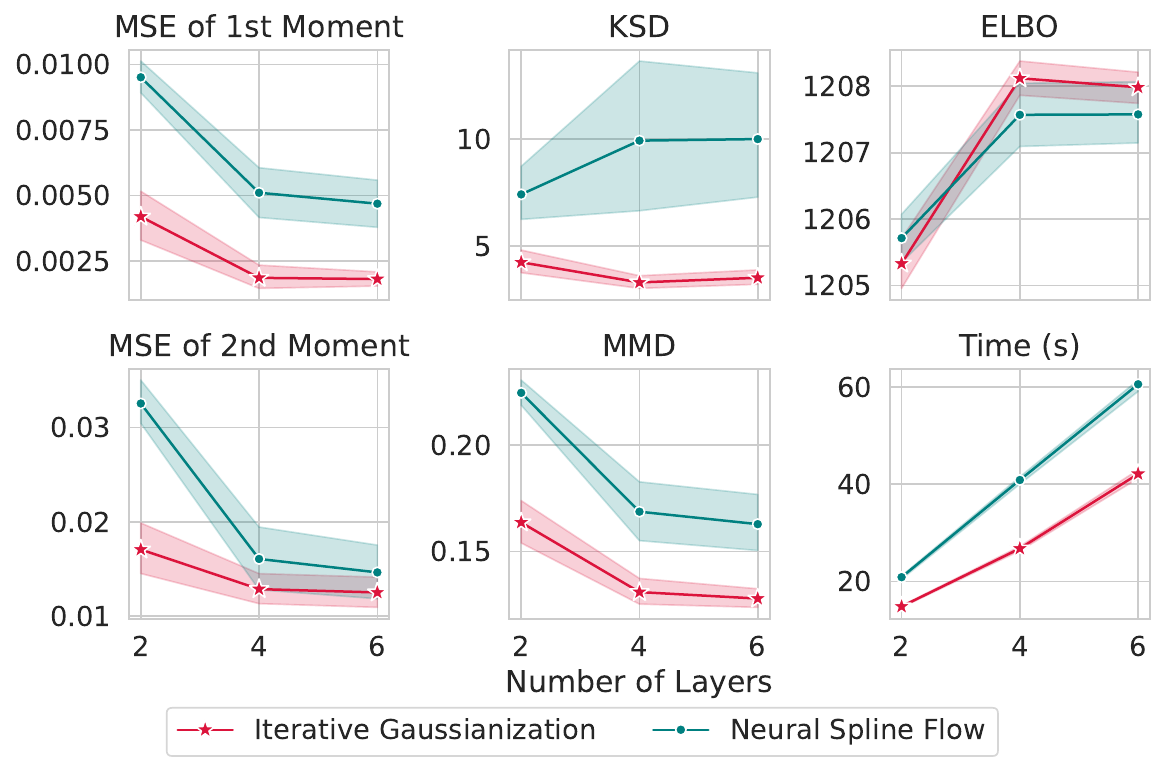}
    \caption{Average performance metrics for the sparse logistic regression experiment, with error bars showing variation across 20 independent replicates.}
    \label{fig: sparse logistic}
\end{figure}

\subsection{Item response theory model}
\label{sec: irt}

Item response theory (IRT) is a statistical framework for analyzing responses to test items, such as exam questions or survey questionnaires. The model assumes that the probability of a correct response depends on the latent ability of the individual, the difficulty of the item, and the discrimination parameter of the item. We consider the logistic-link IRT model as defined in~\citep{stan_irt}. Specifically, for individual $i\in[I]$ and item $j\in[J]$, the correctness of the response $y_{ij}$ is modeled as
\begin{align*}
    y_{ij}\mid \theta_i, a_j, b_j \sim \mathrm{Bernoulli}(S(a_j(\theta_i - b_j) )),
\end{align*}
where $\theta_i$ is the ability of individual $i$, $a_j$ is the discrimination parameter of item $j$, and $b_j$ is the difficulty parameter of item $j$. The prior distributions are given by
\begin{align*}
    a_j &\sim \mathrm{LogNormal}(0,\, \sigma_a^2),\quad j\in[J], \quad \sigma_a \sim \mathrm{LogNormal}(0,\, 2^2), \\
    b_j &\sim \N(\mu_b, \, \sigma_b^2), \quad j\in[J], \quad \mu_b \sim \N(0,\, 5^2),\quad \sigma_b\sim \mathrm{LogNormal}(0,\,  2^2),\\
    \theta_i &\sim \N(0,\, 1),\quad i\in[I].
\end{align*}
We use the benchmark dataset from Stan\footnote{\url{https://github.com/stan-dev/stat_comp_benchmarks/tree/master/benchmarks/irt_2pl}} with $I=100$ and $J=20$, resulting in a posterior distribution of dimension 143.

Due to the high dimensionality of the target distribution, we use the sliced MMD and sliced Wasserstein-2 distance ($W_2$) to evaluate the performance of the learned variational distribution. We consider the first two principal components (PC-1, PC-2) of the reference samples, as well as the last two principal components (PC-142, PC-143). These PCs correspond to the directions in which the target distribution has the largest and smallest variance, respectively. We compute the sliced MMD and sliced $W_2$ along these directions, using 2000 samples from the learned variational distribution and 2000 reference samples. 

The NSF uses the same architecture as in Section~\ref{sec: sparse logistic}, with 6 layers, and is trained with Adam (learning rate 0.01) for 1000 iterations. The proposed method uses 4 iterations with PCA-based rotations that include all principal components. Each MFVI problem is optimized with Adam (learning rate 0.01) for 200 iterations.

The results are shown in Table~\ref{tab: irt-2pl}. Values in parentheses are standard deviations across 20 independent replicates. We observe that the proposed method consistently outperforms NSF across all metrics, demonstrating its ability to approximate the target distribution more accurately along both high-variance and low-variance directions.

\begin{table}
    \centering
    \begin{tabular}{p{2.7cm}p{2.7cm}p{2.7cm}p{2.7cm}p{2.7cm}}
    \toprule
     & MMD (PC-1) & MMD (PC-2) & MMD (PC-142) & MMD (PC-143) \\
    \midrule
    NSF & 0.799 (0.023) & 0.853 (0.021) & 0.224 (0.022) & 0.326 (0.024) \\
    Gaussianization & \textbf{0.312} (0.021) & \textbf{0.310} (0.035) & \textbf{0.180} (0.018) & \textbf{0.115} (0.047) \\
    \bottomrule
    \toprule
     & $W_2$ (PC-1) \phantom{M} & $W_2$ (PC-2) \phantom{M} & $W_2$ (PC-142) \phantom{M} & $W_2$ (PC-143) \phantom{M} \\
    \midrule
    NSF & 2.022 (0.073) & 2.120 (0.064) & 0.096 (0.026) & 0.069 (0.011) \\
    Gaussianization & \textbf{0.714} (0.032) & \textbf{0.667} (0.057) & \textbf{0.042} (0.005) & \textbf{0.016} (0.006) \\
    \bottomrule
    \end{tabular}
\caption{\label{tab: irt-2pl} Results for the item response theory model. Sliced MMD and sliced $W_2$ in the principal component directions are reported, with standard deviation in parentheses.}
\end{table}

\section{Conclusions}
\label{sec: conclusion}

We proposed an algorithm for sampling from unnormalized densities by iteratively performing MFVI in carefully rotated coordinate systems. The rotations are designed to reduce coordinate dependence and expose non-Gaussian directions, thereby increasing the improvement achieved by MFVI. The rotation-selection step is implemented through a PCA procedure based on the relative scores of the target distribution; it is fast to compute, cheap to store, and allows the method to capture complex distributions with minimal parametrization. The resulting approach provides a computationally efficient alternative to heavily parametrized normalizing flows. We established stationarity guarantees for the iterative procedure, derived a formal gradient-flow structure, and demonstrated its effectiveness on posterior sampling tasks.

In practice, the transport map obtained by the algorithm is approximate, meaning that the approximation $q^{(k)}$ defined in Equation~\eqref{equ: qk} does not exactly match the target. The quality of the approximation can be evaluated using diagnostics such as effective sample size, since the density of $q^{(k)}$ is available in closed form. If the approximation is insufficient, additional iterations can be added without modifying earlier ones. When estimating posterior expectations, variance reduction can be achieved with quasi-Monte Carlo methods~\citep{liu2024transport}. Any residual bias can be corrected, either by using $q^{(k)}$ as a proposal in importance sampling or by applying Metropolis-Hastings corrections in MCMC~\citep{gabrie2022adaptive}. The Gaussianization transformation can also serve as a pre-conditioning step to make the geometry of the target more amenable to MCMC~\citep{parno2018transport,hoffman2019neutra}.

\section*{Acknowledgements}

We thank Bob Carpenter, David Dunson, Aram-Alexandre Pooladian, Lawrence Saul, and Surya Tokdar for their valuable discussions and feedback.

\appendix

\section{Proof of Theorem~\ref{thm: gaussian convergence}}
\label{sec: proof of gaussian convergence}

\begin{proof}
    Note that the MF approximation of $p^{(k)}$ ($k\geq1$) is $\gamma=\N(0,I_d)$, so we can write $p^{(k)}=\N(0,\Sigma)$, where $(\Sigma^{-1})_{ii}=1$ for $1\leq i\leq d$. Write the eigendecomposition of $\Sigma$ as $\Sigma=U\Lambda U\tran$, where $\Lambda=\diag(\lambda_1,\ldots,\lambda_d)$. Denote $\nu_i=1/\lambda_i$, then they must satisfy 
    \begin{align*}
        \sum_{i=1}^d \nu_i = \sum_{i=1}^d \frac{1}{\lambda_i} = d.
    \end{align*}
    The KL divergence at iteration $k$ is equal to
    \begin{align*}
        \calL_k=\kl{\N(0,I_d)}{\N(0,\Sigma)} = \frac12 \sum_{i=1}^d (\log \lambda_i + \frac{1}{\lambda_i} -1)=\frac12 \sum_{i=1}^d \log \lambda_i.
    \end{align*}
    For an orthogonal matrix $R$, the MF approximation of $\N(0, R\Sigma R\tran)$ is $\N(0, S)$, where $S$ is diagonal with $S_{ii}^{-1}=((R\Sigma R\tran)^{-1})_{ii} $. Write $\tilde R=R U$, then 
    \begin{align*}
        \frac{1}{S_{ii}} = (\tR \Lambda^{-1} \tR\tran )_{ii} = \sum_{j=1}^d \frac{\tR_{ij}^2}{\lambda_j}.
    \end{align*}
    The KL divergence at iteration $k+1$ is
    \begin{align*}
        \calL_{k+1}
        &=\kl{\N(0,S)}{\N(0,R \Sigma R\tran)}\\
        &=\frac12\Big[\tr(S(R\Sigma R\tran)^{-1} ) + \log|\Sigma| - \log|S| - d  \Big]\\
        &=\frac12\Big[\tr(S\tR\Lambda^{-1} \tR\tran ) + \sum_{i=1}^d \log\lambda_i - \sum_{i=1}^d \log S_{ii} - d  \Big].
    \end{align*}
    Note that
    \begin{align*}
        \tr(S\tR\Lambda^{-1} \tR\tran ) &= \sum_{i=1}^d S_{ii} \sum_{j=1}^d \frac{\tR_{ij}^2}{\lambda_j}  = \sum_{i=1}^d S_{ii}\cdot \frac{1}{S_{ii}} = d.
    \end{align*}
   Thus
    \begin{align*}
        \calL_{k+1} &= \frac12\Big[\sum_{i=1}^d \log\lambda_i - \sum_{i=1}^d \log S_{ii} \Big] 
        = \frac12\Big[\sum_{i=1}^d\log\lambda_i + \sum_{i=1}^d \log \sum_{j=1}^d \frac{\tilde R_{ij}^2}{\lambda_j} \Big].
    \end{align*}
    Note that $\tR=RU$ is a uniformly distributed orthogonal matrix.
    Applying Lemma~\ref{lem: bound dirichlet sum} with $\nu_i=1/\lambda_i$, we have
    \begin{align*}
        \bbE_{R}\Big[\sum_{i=1}^d\log\sum_{j=1}^d \tilde R_{ij}^2\nu_j \Big] \leq \frac{2}{(d+2)\chi^2} \sum_{i=1}^d \log\nu_i.
    \end{align*}
    Therefore,
    \begin{align*}
        \EE[R]{\calL_{k+1}}&\leq \frac12\left[\sum_{i=1}^d \log\lambda_i - \frac{2}{(d+2)\chi^2} \sum_{i=1}^d\log\lambda_i \right]\\
        &=\frac12 \Big(1-\frac{2}{(d+2)\chi^2} \Big)\sum_{i=1}^d \log\lambda_i\\
        &=\Big(1-\frac{2}{(d+2)\chi^2} \Big)\cdot \calL_k.
    \end{align*}
    This proves the theorem.
\end{proof}

\begin{rmk}
A slight modification of the proof of Lemma~\ref{lem: bound dirichlet sum} gives  
\begin{align*}
    \bbE_{R}\Big[\sum_{i=1}^d\log\sum_{j=1}^d \tilde R_{ij}^2\nu_j \Big] \geq \frac{2\chi^2}{(d+2)} \sum_{i=1}^d \log\nu_i.
\end{align*}
This gives the lower bound $\bbE\calL_{k+1}\geq (1-\frac{2\chi^2}{d+2})\calL_k $ when $\frac{2\chi^2}{d+2}< 1$.
\end{rmk}

\begin{lem}\label{lem: bound dirichlet sum}
    Suppose $\nu_i>0$ and $\sum_{i=1}^d \nu_i=d$. For $\nu_{\max}=\underset{1\leq i\leq d}{\max}\, \nu_i $ and $\nu_{\min}=\underset{1\leq i\leq d}{\min}\, \nu_i $, assume $\frac{\nu_{\max}}{\nu_{\min}}\leq\chi$. Then
    \begin{align*}
        \EE[R]{\sum_{i=1}^d \log\sum_{j=1}^d R_{ij}^2\nu_j } \leq \frac{2}{(d+2)\chi^2} \sum_{i=1}^d \log\nu_i,
    \end{align*}
    where the expectation is taken over a uniformly distributed orthogonal matrix $R\in\R^{d\times d}$.
\end{lem}
\begin{proof}
    Note that each row of $R$ is uniformly distributed on $\bbS^{d-1}=\{\bfv\in\R^d: \bfv\tran\bfv=1 \}$.
    Denote $\bfy=(R_{i1}^2,\ldots,R_{id}^2)$, so we have $\sum_{i=1}^d y_i=1$ and $\EE{y_i}=1/d$.
    The expectation on the left-hand side is equal to
    \begin{align*}
        \EE[R]{\sum_{i=1}^d \log\sum_{j=1}^d R_{ij}^2\nu_j } =d\cdot \EE[\bfy]{\log \sum_{i=1}^d\nu_i y_i }.
    \end{align*}
    Define the random variable $W=\sum_{i=1}^d \nu_i y_i $, which satisfies
    \begin{align*}
        \nu_{\min} \leq W \leq \nu_{\max},\quad \EE{W}=\frac1d\sum_{i=1}^d \nu_i = 1.
    \end{align*}
    Applying Lemma~\ref{lem: jensen gap} with $g(x)=-\log x$ and the random variable $W$, we have
    \begin{align*}
        \frac{1}{2\nu_{\max}^2} \Var{W} \leq -\EE{\log W } \leq \frac{1}{2\nu_{\min}^2} \Var{W}.
    \end{align*}

    Similarly, define $Z$ as the random variable that is uniformly distributed on $\{\nu_1,\ldots,\nu_d\}$. Then $Z$ is also bounded between $\nu_{\min}$ and $\nu_{\max}$, and satisfies $\EE{Z}=\frac1d\sum_{i=1}^d \nu_i= 1$. Applying Lemma~\ref{lem: jensen gap} again with $g(x)=-\log x$ and the random variable $Z$, we have
    \begin{align*}
        \frac{1}{2\nu_{\max}^2} \Var{Z} \leq -\EE{\log Z } \leq \frac{1}{2\nu_{\min}^2} \Var{Z}.
    \end{align*}
    By Lemma~\ref{lem: variance of Y}, we have $\Var{W}=\frac{2}{d+2}\Var{Z}$. Therefore,
    \begin{align*}
        \EE{\log W} &\leq -\frac{1}{2\nu_{\max}^2} \Var{W}\\
        &=-\frac{1}{2\nu_{\max}^2} \frac{2}{d+2}\Var{Z}\\
        &\leq \frac{1}{2\nu_{\max}^2} \frac{2}{d+2} \cdot 2\nu_{\min}^2 \EE{\log Z}\\
        &\leq \frac{2}{(d+2)\chi^2} \frac{1}{d}\sum_{i=1}^d \log\nu_i.
    \end{align*}
    This proves that $d\cdot \EE{\log \sum_{i=1}^d\nu_i y_i }\leq \frac{2}{(d+2)\chi^2}\sum_{i=1}^d\log\nu_i $.

\end{proof}

\begin{lem}\label{lem: jensen gap}
    Let $X$ be a random variable that is bounded between $[a,b]$. Let $g:[a,b]\to \R$ be twice continuously differentiable. Then 
    \begin{align*}
        \frac12 \inf_{x\in[a,b]} g''(x)\Var{X}\leq \EE{g(X)} - g(\EE{X}) \leq \frac{1}{2} \sup_{x\in[a,b]}g''(x) \Var{X}.
    \end{align*}
\end{lem}
\begin{proof}[Proof of Lemma~\ref{lem: jensen gap}]
    Denote $\mu=\EE{X}\in[a,b]$ as the mean of $X$. Taking the Taylor expansion of $g(X)$ around $\mu$ gives
    \begin{align*}
        \frac12\inf_{x\in[a,b]} g''(x) (X-\mu)^2\leq g(X) - g(\mu) - g'(\mu)(X-\mu)\leq \frac12 \sup_{x\in[a,b]} g''(x) (X-\mu)^2.
    \end{align*}
    Taking expectation of both sides proves the result.
\end{proof}

\begin{lem}\label{lem: variance of Y}
    Suppose $\bfr$ is uniformly distributed on the sphere $\bbS^{d-1}$ and $y_i=r_i^2$. Suppose $\nu_1,\ldots,\nu_d\geq 0$ and satisfy $\sum_{i=1}^d \nu_i=d$. Then 
    \begin{align*}
    \Var{\sum_{i=1}^d \nu_i y_i } = \frac{2}{(d+2)d}\sum_{i=1}^d (\nu_i-1)^2.
    \end{align*}
\end{lem}
\begin{proof}[Proof of Lemma~\ref{lem: variance of Y}]
    Because $\bfy=(y_1,\ldots,y_d)\sim \mathrm{Dirichlet}(1/2, \ldots, 1/2)$, we have from~\cite[Chapter 49]{kotz2019continuous} that
    \begin{align*}
        \Cov{y_i,y_j}=\frac{\delta_{i,j}\frac1d - \frac{1}{d^2} }{\frac{d}{2}+1 }.
    \end{align*}
    The variance of $\sum_{i=1}^d \nu_i y_i$ can be computed as
    \begin{align*}
    \Var{\sum_{i=1}^d \nu_i y_i}&=\sum_{i=1}^d \Var{\nu_i y_i } + \sum_{i\neq j} \Cov{\nu_i y_i, \nu_j y_j}\\
    &=\sum_{i=1}^d \nu_i^2 \frac{\frac1d - \frac{1}{d^2} }{\frac{d}{2}+1}  + \sum_{i\neq j}\nu_i \nu_j \frac{-\frac{1}{d^2} }{\frac{d}{2}+1 }\\
    &=\sum_{i=1}^d \nu_i^2 \frac{\frac1d}{\frac{d}{2}+1} + \sum_{i,j=1}^d \nu_i \nu_j \frac{-\frac{1}{d^2} }{\frac{d}{2}+1 }\\
    &=\frac{2}{d(d+2)}\sum_{i=1}^d \nu_i^2 -\frac{2}{d^2(d+2)}(\sum_{i=1}^d\nu_i)^2\\
    &=\frac{2}{d(d+2)}\sum_{i=1}^d \nu_i^2 -\frac{2}{d+2}\\
    &=\frac{2}{d+2}(\frac1d\sum_{i=1}^d \nu_i^2-1 )\\
    &=\frac{2}{d+2}\cdot \frac1d\sum_{i=1}^d (\nu_i-1)^2 .
    \end{align*}
\end{proof}

\section{Proofs for Section~\ref{sec: stationarity guarantee}}
\label{sec: proofs}

\subsection{Proof of Theorem~\ref{thm: stationary point}}
\label{prf: stationary point}

\begin{proof}[Proof of Theorem~\ref{thm: stationary point}]
By the standing polynomial-growth assumptions, the log-density ratio
\(r=\log(p/\gamma)\) and its gradient \(h=\nabla r\) have at most polynomial
growth. In particular, \(r\in L^2(\gamma)\) and \(h\in L^2(\gamma;\mathbb R^d)\),
so the argument and in particular the Hermite expansions below are justified.

    For Haar-almost every \(R\), \(\Delta_{\mfvi}(R\#p)=0\). Fix such an \(R\). Let the first column of $R\tran$ be $\theta\in\bbS^{d-1}$ and the remaining columns form $\theta^\perp$.
    Since $\gamma$ satisfies the MFVI optimality condition for $p_R(\bfx)=p(R\tran \bfx)$, we have
    \begin{align*}
        \EE[\gamma_{-1}]{\nabla_{x_1} \log p_R(\bfx)} = \nabla_{x_1} \log\gamma_1(x_1).
    \end{align*}
    Define $h(\bfx) :=\nabla\log (p/\gamma)(\bfx)$ and $h_R(\bfx)=\nabla\log(p_R/\gamma)(\bfx)=R h(R\tran \bfx)$, then the last display is equivalent to $\EE[\gamma_{-1}]{\langle \theta, h(R\tran \bfx)\rangle} = 0$.
    Thus, for any $\gamma$-integrable function $\psi:\R\to\R$, 
    \begin{align*}
        \EE[\gamma]{\langle \theta, h(R\tran \bfx)\rangle \cdot \psi(x_1)} = 0.
    \end{align*}
    By the change of variable $\bfy=R\tran \bfx\sim \gamma$, we have
    \begin{align}\label{equ: opt cond psi}
        \EE[\gamma]{\langle \theta, h(\bfy)\rangle \cdot \psi(\theta\tran \bfy)} = 0.
    \end{align}
    Since the set of such rotations has full Haar measure and the first column of
a Haar-distributed orthogonal matrix is uniformly distributed on \(\bbS^{d-1}\),
\eqref{equ: opt cond psi} holds for spherical-almost every
\(\theta\in\bbS^{d-1}\).

    Conceptually, \eqref{equ: opt cond psi} says that the averaged directional
derivative of \(r=\log(p/\gamma)\) vanishes along almost every direction
\(\theta\). Indeed, writing \(Y=t\theta+Y^\perp\), where
\(Y^\perp\sim \gamma\) on \(\theta^\perp\), it implies
\[
    \frac{\rd}{\rd t}\,
    \EE{r(t\theta+Y^\perp)}=0 .
\]
Hence \(\EE{r(Y)\mid \theta^\top Y}\) is constant for almost every such
\(\theta\). Therefore \(r\) is orthogonal in \(L^2(\gamma)\) to functions of
one-dimensional projections. Since the span of polynomials
\((\theta^\top x)^n\), over \(\theta\in\mathbb S^{d-1}\) and \(n\ge0\), is dense
in \(L^2(\gamma)\), this already gives the desired conclusion. We now give the
corresponding Hermite-coefficient proof, since the same Hermite machinery is
used later in Section~\ref{sec: pca}.

Since $h\in L^2(\gamma)$, we can write the Hermite polynomial expansion of $h_i$ ($1\leq i\leq d$) as
\begin{align}\label{equ: h hermite}
    h_i(\bfx) = \sum_{\bfk\in\bbN^d} c_{i,\bfk} \cdot \he_{\bfk}(\bfx),
\end{align}
where the sum is over all the multi-indices in $\bbN^d$ and $\he_{\bfk}(\bfx)=\prod_{i=1}^d \he_{k_i}(x_i)$ is the product of univariate Hermite polynomials. We use the normalized probabilist's Hermite polynomials, which satisfy
\begin{align*}
    \EE[\gamma]{\he_{\bfk}(\bfx)\cdot \he_{\bfk'}(\bfx) } = \delta_{\bfk,\bfk'}.
\end{align*}
For $\bfk\in\bbN^d$, let $|\bfk|=\sum_{i=1}^d k_i$, $\bfk!=\prod_{i=1}^d k_i!$, and $\theta^{\bfk}=\prod_{i=1}^d \theta_i^{k_i}$.

Substituting the Hermite expansion of $h$ into Equation~\eqref{equ: opt cond psi} and setting  $\psi=\he_{m}$ gives
\begin{align*}
    0 &= \EE[\gamma]{\langle \theta, h(\bfx)\rangle \cdot \he_m(\theta\tran \bfx) }\\
    &=\sum_{i=1}^d \theta_i \sum_{\bfk\in\bbN^d} c_{i,\bfk} \cdot \EE[\gamma]{\he_{\bfk}(\bfx)\cdot \he_m(\theta\tran \bfx)}.
\end{align*}
By the addition formula of Hermite polynomials~\citep[8.958]{gradshteyn2014table} 
\begin{align*}
    \he_{m}(\theta\tran \bfx) = \sum_{|\bfj| = m } \sqrt{\frac{m! }{\bfj!} } \, \theta^{\bfj}\cdot \he_{\bfj}(\bfx),
\end{align*}
the last equation becomes
\begin{align*}
    0 &=\sum_{i=1}^d \theta_i \sum_{\bfk\in\bbN^d} c_{i,\bfk} \sum_{|\bfj|=m}\sqrt{\frac{m!}{\bfj!} } \theta^{\bfj} \cdot \EE[\gamma]{\he_{\bfk}(\bfx)\cdot \he_{\bfj}(\bfx) }.
\end{align*}
Since the expectation vanishes unless $\bfk=\bfj$, it reduces to
\begin{align*}
    0=\sum_{i=1}^d \theta_i \sum_{\bfk: |\bfk|=m} c_{i,\bfk} \sqrt{\frac{m!}{\bfk!} } \theta^{\bfk}
    =\sum_{\bfk: |\bfk|=m} \sqrt{\frac{m!}{\bfk!} } \cdot\sum_{i=1}^d c_{i,\bfk} \theta_i \cdot \theta^{\bfk}.
\end{align*}

Write the Hermite expansion of \(r=\log(p/\gamma)\) as
\[
    r(\bfx)=\sum_{\bfk\in\bbN^d} b_{\bfk}\he_{\bfk}(\bfx).
\]
By Gaussian integration by parts and
\(\partial_i\he_{\bfk}=\sqrt{k_i}\he_{\bfk-\bfe_i}\), the coefficients satisfy
\[
    c_{i,\bfk-\bfe_i}
    =
    \EE[\gamma]{\partial_i r(\bfx)\he_{\bfk-\bfe_i}(\bfx)}
    =
    \sqrt{k_i}\,\EE[\gamma]{r(\bfx)\he_{\bfk}(\bfx)}
    =
    \sqrt{k_i}\,b_{\bfk},
\]
with the convention that this term is zero when \(k_i=0\). Substituting this
identity into the preceding display gives
\[
    0
    =
    (m+1)\sum_{|\bfk|=m+1}
    \sqrt{\frac{m!}{\bfk!}}\,
    b_{\bfk}\theta^{\bfk}
\]
for almost every \(\theta\in\mathbb S^{d-1}\). Since this homogeneous polynomial
vanishes on a set of full spherical measure, it vanishes identically. Therefore
\(b_{\bfk}=0\) for every \(|\bfk|=m+1\). Hence the degree-\(m\) Hermite component
of \(h=\nabla r\) vanishes. Repeating this for all \(m\ge0\) gives \(h=0\), and
therefore \(p=\gamma\).

\end{proof}

\subsection{Proof of Proposition~\ref{prop: stationarity guarantee}}
\label{prf: stationarity guarantee}

\begin{proof}
The first part follows from Theorem~\ref{thm: stationary point}.

The algorithm implies $\Delta_{\mfvi}(R_k\# p^{(k-1)})=\kl{\gamma}{p^{(k-1)}} - \kl{\gamma}{p^{(k)}}$. Summing over $k=1,\ldots,K$, 
\begin{align*}
    \sum_{k=1}^K \Delta_{\mfvi}(R_k\# p^{(k-1)}) = \kl{\gamma}{p^{(0)}} - \kl{\gamma}{p^{(K)}} \leq \kl{\gamma}{p}.
\end{align*}
When the rotations are chosen deterministically and greedily, we have $\widebar \Delta_{\mfvi}(p^{(k-1)})=\Delta_{\mfvi}(R_k\# p^{(k-1)}) $, thus $\sum_{k=1}^K \widebar \Delta_{\mfvi}(p^{(k-1)}) \leq  \kl{\gamma}{p}$, which implies that $\widebar\Delta_{\mfvi}(p^{(K)} ) \to 0$ as $K\to\infty$, and $\min_{1\leq k\leq K}\widebar \Delta_{\mfvi}(p^{(k-1)}) = O(\frac{1}{K})$.

When $R_k$ are random, we denote $\delta_{k-1}=\Delta_{\mfvi}(R_k\# p^{(k-1)})$, and write
\begin{align*}
    \widebar\Delta_{\mfvi}(p^{(k-1)}) =\EE[R_k]{\delta_{k-1}}=: \bar\delta_{k-1},
\end{align*}
where the expectation is taken over the randomness in $R_k$; equivalently, $\mathbb{E}_{R_k}[\delta_{k-1}] = \mathbb{E}[\delta_{k-1}|\sigma(R_1,\ldots, R_{k-1})]$ as conditional expectation. 

Since $\sum_{k=1}^K\delta_{k-1}\leq \kl{\gamma}{p}$, taking expectation yields
\begin{align*}
    \EE{\sum_{k=1}^K \bar\delta_{k-1} } = \EE{\sum_{k=1}^K \delta_{k-1} } \leq \kl{\gamma}{p}.
\end{align*}
Since $\bar\delta_{k-1}\geq 0$, by the monotone convergence theorem, $\sum_{k=1}^\infty \bar\delta_{k-1} < \infty$ almost surely, which implies $\bar\delta_{K}\stackrel{a.s.}{\to} 0$ as $K\to\infty$. Moreover, $\EE{\min_{1\leq k\leq K} \bar\delta_{k-1}} \leq \EE{\frac1K \sum_{k=1}^K\bar\delta_{k-1} } = O(\frac{1}{K})$.

\end{proof}

\section{Proofs for Section~\ref{sec: instantaneous}}

\subsection{Proof of Proposition~\ref{prop: pFI gaussian}}
\label{prf: pFI gaussian}
\begin{proof}
The projected FI is
\begin{align*}
    \tI(\gamma, p) = \sum_{i=1}^d \EE{(\EE{\nabla_i \log(p/\gamma)(\bfx)\mid x_i })^2 }=
    \sum_{i=1}^d \EE{ ((1-\omega_i)x_i)^2 }=\sum_{i=1}^d (1-\omega_i)^2,
\end{align*}
where $\omega_i = \Omega_{ii}$.

The MF approximation of $p$ is $q=\N(0,D^{-1})$, where $D=\diag(\Omega)$. We have
\begin{align*}
    \kl{q}{p} &= \frac12 \Big(\tr(\Omega D^{-1}) - d + \log\frac{|\Omega^{-1}|}{|D^{-1}|} \Big) = \frac12 \Big(-\log|\Omega| + \log|D|  \Big)\\
    \kl{\gamma}{p} &= \frac12 \Big(\tr(\Omega) - d - \log|\Omega| \Big).
\end{align*}
Thus,
\begin{align*}
    \Delta_{\mfvi} = \kl{\gamma}{p} - \kl{q}{p} = \frac12 \Big(\tr(\Omega) - d -  \log|D|  \Big)=\frac12\sum_{i=1}^d (\omega_i - \log\omega_i-1 ).
\end{align*}
Let $f(x)=x-\log x-1$. We have
\begin{align*}
    f(x)=f(1)+f'(1)(x-1) + \int_1^x f''(t) (x-t) \rd t = \int_1^x \frac{1}{t^2} (x-t) \rd t.
\end{align*}
Apply $t=1+\tau(x-1)$, we have
\begin{align*}
    f(x)=\int_0^1 \frac{1-\tau}{(1+\tau(x-1))^2} (x-1)^2 \rd \tau.
\end{align*}
Thus,
\begin{align*}
    \frac{x-\log x - 1}{(x-1)^2} = \int_0^1 \frac{1-\tau}{(1+\tau(x-1))^2} \rd \tau.
\end{align*}
Since $1+\tau(x-1)\leq \max(x, 1)$ for $\tau\in[0,1]$, 
\begin{align*}
    \frac{x-\log x - 1}{(x-1)^2} \geq \int_0^1 \frac{1-\tau}{\max(x,1)^2} \rd \tau = \frac{1}{2\max(x,1)^2}.
\end{align*}
Similarly,
\begin{align*}
    \frac{x-\log x - 1}{(x-1)^2} \leq  \frac{1}{2\min(x,1)^2}.
\end{align*}

Since $C_1\leq \omega_i\leq C_2$ for $1\leq i\leq d$, we have
\begin{align*}
    \frac{1}{4\max(C_2, 1)^2} \tI(\gamma,p)\leq \Delta_{\mfvi} \leq \frac{1}{4\min(C_1,1)^2} \tI(\gamma, p).
\end{align*}
\end{proof}

\subsection{Instantaneous analysis}\label{sec: instantaneous analysis}

\paragraph{Continuity equation.}
If $q_\ep=(\mathrm{id} + \ep \bfv)\# q$, then the continuity equation $\partial_\ep q_\ep\big\vert_{\ep=0} + \nabla\cdot(q \bfv)=0$ holds in the distributional sense.
To see this, take any smooth test function $\phi:\R^d\to\R$. Then
\begin{align*}
    \int \phi(\bfy) q_\ep(\bfy) \rd\bfy &= \int \phi(\bfx+\ep \bfv(\bfx)) q(\bfx) \rd\bfx .
\end{align*}
Differentiating both sides with respect to $\ep$ and evaluating at $\ep=0$ gives
\begin{align*}
    \int \phi(\bfy) \partial_\ep q_\ep\;\Big\vert_{\ep=0} \rd\bfy &= \int \langle \nabla\phi(\bfx), \bfv(\bfx)\rangle q(\bfx) \rd\bfx.
\end{align*}
Integrating by parts on the right-hand side gives
\begin{align*}
    \int \phi(\bfy) \partial_\ep q_\ep\;\Big\vert_{\ep=0} \rd\bfy &= -\int \phi(\bfx) \nabla\cdot(q \bfv)(\bfx) \rd\bfx.
\end{align*}
Since this holds for any test function $\phi$, we have $\partial_\ep q_\ep \big\vert_{\ep=0} = -\nabla\cdot(q \bfv)$.

\paragraph{Instantaneous rate of $\kl{q_\ep}{p}$ (Equation~\eqref{equ: kl instantaneous rate}).}
Note that
\begin{align*}  
    \partial_\ep \kl{q_\ep}{p}\;\Big\vert_{\ep=0} &=\partial_\ep \int q_\ep\log\frac{q_\ep}{p}\; \Big\vert_{\ep=0} = \int \Big(\log\frac{\gamma}{p} + 1\Big) \partial_\ep q_\ep\;\Big\vert_{\ep=0}\\
    &=- \int \nabla\cdot(\gamma \bfv) \Big(\log\frac{\gamma}{p} + 1\Big),
\end{align*}
where we have used $\frac{\rd}{\rd q} q\log(q/p) = \log(q/p) + 1$ and the continuity equation. Integrating by parts on the right-hand side gives
\begin{align*}
    \partial_\ep \kl{q_\ep}{p}\;\Big\vert_{\ep=0} &=\int \gamma \bfv\cdot \nabla\Big(\log\frac{\gamma}{p} + 1\Big)  =-\EE[\gamma]{ \big\langle \bfv(\bfx), \nabla\log \frac{p}{\gamma}(\bfx) \big\rangle}.
\end{align*}

\paragraph{Derivation of Equation~\eqref{equ: optimal vector field}.}

The steepest descent direction is the solution to the following optimization problem:
\begin{align*}
&\max_{\bfv}\; \EE[\gamma]{\langle \bfv(\bfx), h(\bfx)\rangle} - \frac12 \|\bfv\|_{L^2(\gamma)}^2 \\
&\text{subject to}\; \text{$\bfv$ is coordinatewise }
\end{align*}
where we denote $h=\nabla\log(p/\gamma)$. Since $\bfv(\bfx)=(v_1(x_1), \ldots, v_d(x_d))$, we have
\begin{align*}
    \EE[\gamma]{\sum_{i=1}^d v_i(x_i) h_i(\bfx)} = \EE[\gamma]{\sum_{i=1}^d v_i(x_i)\EE{h_i(\bfx)\mid x_i} }=\EE[\gamma]{\langle \bfv(\bfx), \bar h(\bfx)\rangle },
\end{align*}
where $\bar h_i(x_i) = \EE{h_i(\bfx)\mid x_i}$. Therefore, the optimal $\bfv$ is $\bar h$.

\subsection{Constrained and weighted Wasserstein gradient flow}
\label{sec: gradient flow derivation}

We give the formal construction underlying Section~\ref{sec: gradient flow}.
Let
\[
    \calE(p)=\kl{\gamma}{p}.
\]
Throughout this subsection, we assume that all densities are positive and sufficiently
smooth, and that the integrations by parts below are justified.

For a distribution $p$, let $R_p$ denote the rotation chosen at $p$, and write
$\tp=R_p\#p$ for the rotated target. Define the operator
\[
    \calR_p \bfv(\bfx)=R_p\tran \bfv(R_p\bfx).
\]
The admissible velocity fields are rotated coordinatewise gradient fields. We therefore
define the constrained tangent space
\[
    \tT_p\calP_2
    =
    \Big\{
    \calR_p\nabla\phi:
    \phi(\bfy)=\sum_{i=1}^d\phi_i(y_i),\;
    \phi_i\in C_c^\infty(\R)
    \Big\}.
\]
We equip this space with the metric
\[
    \tg_p(\calR_p\nabla\phi_1,\calR_p\nabla\phi_2)
    =
    \EE[\bfy\sim\tp]{
    \left\langle
    \nabla\phi_1(\bfy),
    \sfM_p(\bfy)^{-1}\nabla\phi_2(\bfy)
    \right\rangle
    },
\]
where
\[
    \sfM_p(\bfy)
    =
    \diag\left(
    \frac{\tp_1(y_1)}{\gamma_1(y_1)},\ldots,
    \frac{\tp_d(y_d)}{\gamma_d(y_d)}
    \right),
\]
and $\tp_i$ and $\gamma_i$ denote the $i$-th marginal densities of $\tp$ and $\gamma$.
The weighting by $\sfM_p^{-1}$ arises because the local perturbation in
Section~\ref{sec: projected FI} is measured in $L^2(\gamma)$, whereas the
target-side flow is written in terms of the current target distribution.

We now derive the Wasserstein gradient of $\calE$ in this constrained and weighted tangent space.
For a smooth function $\varphi:\R^d\to\R$, define $\calA_p\nabla\varphi=\nabla\phi$,
where $\phi(\bfy)=\sum_{i=1}^d\phi_i(y_i)$ is defined as
\[
    \partial_i\phi_i(y_i)
    =
    \frac{\tp_i(y_i)}{\gamma_i(y_i)}
    \EE[\tp]{
    \partial_{y_i}\left[\varphi(R_p\tran\bfy)\right]\mid y_i
    }.
\]
Thus, $\calA_p$ maps a general gradient field to a coordinatewise gradient field in the
rotated coordinate system, with the marginal reweighting induced by the metric
$\tg_p$.

Let $p_t$ be a smooth curve satisfying the continuity equation
\[
    \partial_t p_t+\nabla\cdot(p_t\bfv_t)=0,
    \qquad
    \bfv_t\in\tT_{p_t}\calP_2.
\]
Then
\[
    \frac{\rd}{\rd t}\calE(p_t)
    =
    \EE[p_t]{
    \left\langle
    \nabla\left(-\frac{\gamma}{p_t}\right),
    \bfv_t
    \right\rangle
    }.
\]
Suppose the velocity field is in the constrained tangent space, i.e. $\bfv_t=\calR_{p_t}\nabla\phi \in \tilde{T}_p \calP_2$. Changing the variable
$\bfy=R_{p_t}\bfx$ and using the tower property of conditional expectation gives
\[
\begin{aligned}
    \frac{\rd}{\rd t}\calE(p_t)
    &=
    \EE[\tp_t]{
    \left\langle
    R_{p_t}\nabla\left(-\frac{\gamma}{p_t}\right)(R_{p_t}\tran\bfy),
    \nabla\phi(\bfy)
    \right\rangle
    }                                                   \\
    &=
    \tg_{p_t}
    \left(
    \calR_{p_t}\calA_{p_t}\nabla\left(-\frac{\gamma}{p_t}\right),
    \calR_{p_t}\nabla\phi
    \right).
\end{aligned}
\]
Therefore, the constrained and weighted Wasserstein gradient of $\calE$ at $p$ is
\[
    \widetilde\nabla \calE(p)
    =
    \calR_p\calA_p\nabla\left(-\frac{\gamma}{p}\right),
\]
and the corresponding gradient flow is
\[
    \partial_t p_t
    =
    \nabla\cdot
    \left(
    p_t\,
    \calR_{p_t}\calA_{p_t}
    \nabla\left(-\frac{\gamma}{p_t}\right)
    \right).
\]

It remains to identify the energy dissipation. Since $\gamma$ is invariant under
rotations,
\[
    \left(-\frac{\gamma}{p_t}\right)(R_{p_t}\tran\bfy)
    =
    -\frac{\gamma(\bfy)}{\tp_t(\bfy)}.
\]
Moreover,
\[
    \partial_{y_i}\left(-\frac{\gamma}{\tp_t}\right)(\bfy)
    =
    \frac{\gamma(\bfy)}{\tp_t(\bfy)}
    \partial_{y_i}\log\frac{\tp_t}{\gamma}(\bfy).
\]
It follows from the definition of $\calA_{p_t}$ that
\[
    \left[
    \calA_{p_t}\nabla\left(-\frac{\gamma}{p_t}\right)
    \right]_i(y_i)
    =
    \EE[\gamma]{
    \partial_{y_i}\log\frac{\tp_t}{\gamma}(\bfy)
    \mid y_i
    }.
\]
Therefore,
\begin{align*}
    -\left\|
    \widetilde\nabla\calE(p_t)
    \right\|_{\tg_{p_t}}^2 &= -\tg_{p_t}\left( \calR_{p_t} \calA_{p_t} \nabla(-\frac{\gamma}{p_t}), \calR_{p_t} \calA_{p_t} \nabla(-\frac{\gamma}{p_t}) \right) \\
    &=-\EE[\tp_t]{\left\langle \calA_{p_t} \nabla(-\frac{\gamma}{p_t}), \mathsf{M}_{p_t}^{-1} \calA_{p_t} \nabla(-\frac{\gamma}{p_t}) \right\rangle}\\
    &=-\sum_{i=1}^d \EE[\tp_t]{\frac{\gamma_i(y_i)}{\tp_{t,i}(y_i)} \Bigl(\EE[\tp_t]{ \frac{\tp_{t,i}(y_i)}{\gamma_i(y_i)} \partial_{y_i}(-\frac{\gamma}{\tp_t})(\bfy) \mid y_i}\Bigr)^2}\\
    &=-\sum_{i=1}^d \EE[\gamma_i]{\Bigl(\EE[\tp_t]{\frac{\tp_{t,i}}{\gamma_i}(\frac{\gamma}{\tp_t})\partial_{y_i}(\log\frac{\tp_t}{\gamma}) (\bfy)  \mid y_i} \Bigr)^2 }\\
    &=-\sum_{i=1}^d \EE[\gamma_i]{\Bigl(\EE[\gamma]{\partial_{y_i}\log\frac{\tp_t}{\gamma}( \bfy) \mid y_i }\Bigr)^2 }\\
    &= -\tI(\gamma, \tp_t).
\end{align*}
Consequently, along the constrained and weighted gradient flow,
\[
    \frac{\rd}{\rd t}\kl{\gamma}{p_t}
    =
    -\left\|
    \widetilde\nabla\calE(p_t)
    \right\|_{\tg_{p_t}}^2
    =
    -\tI(\gamma,R_{p_t}\#p_t).
\]

\subsection{Discussion on the Polyak--{\L}ojasiewicz (PL) inequality}
\label{LSI discussion}
Along this constrained and weighted Wasserstein gradient flow, if the
PL-type inequality
$\tI(\gamma,R_p\# p) \geq c\,\kl{\gamma}{p}$
holds for some constant $c>0$, then the KL divergence decays exponentially
fast along the flow. Such an inequality does not hold in general. However, it
does hold under the following two conditions: first, a log-Sobolev-type
inequality along the flow,
$\calD_F(\gamma,p_t)\geq c_1\kl{\gamma}{p_t}$ for $t\geq0$;
and second, a lower bound on the projected Fisher information in terms of the
Fisher divergence,
$\tI(\gamma,p_t)\geq c_2\calD_F(\gamma,p_t)$.
Together, these imply the PL inequality with $c=c_1c_2$.

In the following, we provide a lower bound on the average projected FI $\bar I(\gamma, p)=\EE[R]{\tI(\gamma, R\# p) }$ in terms of the Fisher divergence. The rotation is assumed to be drawn uniformly at random.

Suppose \(r=\log(p/\gamma)\in L^2(\gamma)\) and
\(h=\nabla\log(p/\gamma)\in L^2(\gamma)\), and let \(h\) have the Hermite
expansion given in Equation~\eqref{equ: h hermite}.
Let $h^{(m)}_{i}(\bfx) = \sum_{\bfk\in\bbN^d:|\bfk|=m} c_{i,\bfk} \he_\bfk(\bfx)$.

\begin{thm}\label{thm: average projected FI lower bound}
    Suppose that there exists $M\geq0$ and $\delta\in(0,1)$ such that
    \begin{align*}
        \sum_{m=0}^M\|h^{(m)} \|_{L^2(\gamma)}^2 \geq (1-\delta )\cdot \|h\|_{L^2(\gamma)}^2. 
    \end{align*}
    Then there exists $C=C(d,\delta, M)>0$ such that
    \begin{align*}
        \bar I(\gamma, p) \geq C \cdot \calD_F(\gamma, p).
    \end{align*}
\end{thm}
\begin{proof}
From the proof in~\ref{prf: pFI decomposition}, we have $\tI(\gamma, p_R) =  \sum_{m\geq 0} \sum_{i=1}^d \big(a^{(m)}(R_i)\big)^2$, 
where
\begin{align*}
    a^{(m)}(\theta)&=\sum_{\bfk: |\bfk|=m+1}\sum_{i=1}^d \sqrt{\frac{m!}{(\bfk-\bfe_i)!} } c_{i,\bfk-\bfe_i} \cdot \theta^{\bfk}.
\end{align*}
Taking the average over uniformly random $R$, we have
\begin{align*}
    \bar I(\gamma, p) &=d\cdot \sum_{m\geq 0} \EE[\theta]{\big(a^{(m)}(\theta)\big)^2 },
\end{align*}
where the expectation is over $\theta\sim \mathrm{Unif}(\bbS^{d-1})$.

Write the degree-(m+1) Hermite component of $r=\log(p/\gamma)$ as
\[
    r^{(m+1)}(x)
    =
    \sum_{|\mathbf j|=m+1}
    b_{\mathbf j}\he_{\mathbf j}(x).
\]
Then we have $c_{i,\mathbf j-\bfe_i} = b_{\mathbf j}\sqrt{j_i}$, with the convention that the term is zero when \(j_i=0\). 
Hence, we have
\[
    \|h^{(m)}\|_{L^2(\gamma)}^2
    =
    \sum_{i=1}^d
    \sum_{|\mathbf j|=m+1}
    b_{\mathbf j}^2 j_i
    =
    (m+1)\sum_{|\mathbf j|=m+1} b_{\mathbf j}^2,
\]
and
\[
\begin{aligned}
    a^{(m)}(\theta)
    &=
    \sqrt{m+1}
    \sum_{|\mathbf j|=m+1}
    b_{\mathbf j}
    \sqrt{\frac{(m+1)!}{\mathbf j!}}
    \theta^{\mathbf j}.
\end{aligned}
\]

Applying Lemma~\ref{lem: sphereical integral} with \(n=m+1\) and \(u_{\mathbf j}=b_{\mathbf j}\), we obtain
\[
\begin{aligned}
    \mathbb E_\theta[a^{(m)}(\theta)^2]
    &\ge
    (m+1)
    \frac{(m+1)!}{d(d+2)\cdots(d+2m)}
    \sum_{|\mathbf j|=m+1}b_{\mathbf j}^2  \\
    &=
    \frac{(m+1)!}{d(d+2)\cdots(d+2m)}
    \|h^{(m)}\|_{L^2(\gamma)}^2 .
\end{aligned}
\]
Thus
\[
    d\,\mathbb E_\theta[a^{(m)}(\theta)^2]
    \ge
    \kappa_{d,m}
    \|h^{(m)}\|_{L^2(\gamma)}^2,
\]
where
\[
    \kappa_{d,m}
    :=
    \frac{d(m+1)!}{d(d+2)\cdots(d+2m)}
    =
    \frac{(m+1)!}{(d+2)(d+4)\cdots(d+2m)} ,
\]
with the empty product interpreted as \(1\) when \(m=0\).

Consequently,
\[
    \bar I(\gamma,p)
    \ge
    \sum_{m=0}^M
    \kappa_{d,m}
    \|h^{(m)}\|_{L^2(\gamma)}^2 .
\]
Let $\kappa_{d,M}^*:=\min_{0\le m\le M}\kappa_{d,m}$.
By the assumption of the theorem,
\[
    \sum_{m=0}^M
    \|h^{(m)}\|_{L^2(\gamma)}^2
    \ge
    (1-\delta)\|h\|_{L^2(\gamma)}^2.
\]
Since
$\mathcal D_F(\gamma,p)=\|h\|_{L^2(\gamma)}^2$,
we conclude that
\[
    \bar I(\gamma,p)
    \ge
    (1-\delta)\kappa_{d,M}^*
    \mathcal D_F(\gamma,p).
\]
Therefore the theorem holds with $C(d,\delta,M)=(1-\delta)\kappa_{d,M}^*>0$.
\end{proof}

\begin{lem}\label{lem: sphereical integral}
Let \(n\ge 0\) and let
\[
    P(\theta)
    =
    \sum_{|\mathbf j|=n}
    u_{\mathbf j}
    \sqrt{\frac{n!}{\mathbf j!}}\,
    \theta^{\mathbf j},
    \qquad \theta\in \mathbb S^{d-1}.
\]
Then
\[
    \mathbb E_{\theta\sim \mathrm{Unif}(\mathbb S^{d-1})}
    \big[P(\theta)^2\big]
    \ge
    \frac{n!}{d(d+2)\cdots(d+2n-2)}
    \sum_{|\mathbf j|=n}u_{\mathbf j}^2.
\]
\end{lem}
\begin{proof}
Let \(Z\sim N(0,I_d)\). We can write $Z=\rho\Theta$, where
\(\Theta\sim \mathrm{Unif}(\mathbb S^{d-1})\), \(\rho\) is independent of \(\Theta\), and $\rho^2\sim \chi_d^2$.
Since \(P\) is homogeneous of degree \(n\),
$P(Z)=P(\rho\Theta)=\rho^n P(\Theta)$,
and thus
\[
    \mathbb E[P(Z)^2]
    =
    \mathbb E[\rho^{2n}]\cdot \mathbb E[P(\Theta)^2].
\]
Since $\mathbb E[\rho^{2n}]=d(d+2)\cdots(d+2n-2)$, we have
\[
    \mathbb E[P(\Theta)^2]
    =
    \frac{\mathbb E[P(Z)^2]}
    {d(d+2)\cdots(d+2n-2)}.
\]

It remains to lower bound \(\mathbb E[P(Z)^2]\). 
For each multi-index \(\mathbf j\) with \(|\mathbf j|=n\), the normalized Hermite polynomial \(\he_{\mathbf j}\) has leading term
$\frac{Z^{\mathbf j}}{\sqrt{\mathbf j!}}$.
Therefore, we can write
\[
    \sqrt{\frac{n!}{\mathbf j!}}Z^{\mathbf j}
    =
    \sqrt{n!}\,\he_{\mathbf j}(Z)
    +
    q_{\mathbf j}(Z),
\]
where $q_{\mathbf j}$ is a polynomial of degree at most \(n-1\), and thus is orthogonal to $\he_{\bfj}$.

By the orthogonality of Hermite polynomials, we have
\[
\begin{aligned}
    \mathbb E[P(Z)^2]
    &=
    n!\,
    \mathbb E\Big[
        \Big(
        \sum_{|\mathbf j|=n}
        u_{\mathbf j}\he_{\mathbf j}(Z)
        \Big)^2
    \Big]
    +
    \mathbb E\Big[ \Big(\sum_{|\mathbf j|=n} u_{\mathbf j} q_{\mathbf j}(Z)\Big)^2 \Big]  
    \ge
    n!
    \sum_{|\mathbf j|=n}u_{\mathbf j}^2.
\end{aligned}
\]
This proves the claim.
\end{proof}

\section{Proofs for Section~\ref{sec: pca}}

\subsection{Proof of Theorem~\ref{thm: pFI decomposition}}
\label{prf: pFI decomposition}

\begin{proof}
Let $h(\bfx)=\nabla\log(p/\gamma)(\bfx)$, and suppose $h_{i}$ has the Hermite expansion
\[
h_{i}(\bfx) = \sum_{m\geq 0}\; \sum_{\bfk\in\bbN^d:|\bfk|=m} c_{i,\bfk} \he_\bfk(\bfx) .
\]
Let $h^{(m)}_{i}(\bfx) = \sum_{\bfk\in\bbN^d:|\bfk|=m} c_{i,\bfk} \he_\bfk(\bfx)\in\calV^{(m)}$, where $\calV^{(m)} = \big\{\sum_{\bfk:|\bfk|=m}\alpha_{\bfk} \he_{\bfk},\; \alpha_{\bfk}\in\R \big\}$.
Denote $h_R(\bfx)=\nabla\log(p_R/\gamma)(\bfx)=R h(R\tran\bfx)$ and $h_{R,i}^{(m)}(\bfx) = R_i\tran h^{(m)}(R\tran \bfx)$. We have $h_{R,i}^{(m)} \in \calV^{(m)}$ since Hermite polynomials are closed under rotation (e.g. by the addition formula).
We can write
\begin{align*}
    \EE{h_{R,i}(\bfx)\mid x_i} = h_{R,i}^{(0)} + \sum_{m\geq1} \EE{h_{R,i}^{(m)}(\bfx)\mid x_i}.
\end{align*}
Using the orthogonality of Hermite polynomials, for $|\bfk|=m\geq 1$, 
\begin{align*}
    \EE{\he_{\bfk}(\bfx) \mid x_i } = \begin{cases}
        \he_{k_i}(x_i), & \text{if } \bfk_{- i} = 0,\\
        0, & \text{otherwise}.
    \end{cases}
\end{align*}
Hence, for $m\geq 1$, $\EE{h_{R,i}^{(m)}(\bfx)\mid x_i}$ must be proportional to $\he_{m}(x_i)$, so we can write
\begin{align*}
    \EE{h_{R,i}^{(m)}(\bfx)\mid x_i} = a_{R,i}^{(m)} \cdot \he_{m}(x_i),
\end{align*}
where
\begin{align*}
    a_{R,i}^{(m)} &= \EE{h_{R,i}^{(m)}(\bfx) \he_{m}(x_i)} =\EE{R_i\tran h^{(m)}(R\tran \bfx) \he_{m}(x_i) }\\
    &=\EE{R_i\tran h^{(m)}(\bfx) \he_{m}(R_i\tran \bfx) }=: a^{(m)}(R_i).
\end{align*}
Thus
\begin{align*}
    \EE{h_{R,i}(\bfx)\mid x_i} = h_{R,i}^{(0)} + \sum_{m\geq 1} a^{(m)}(R_i) \he_{m}(x_i).
\end{align*}
Computing the second moment and using the orthogonality of Hermite polynomials gives
\begin{align*}
    \EE{\big(\EE{h_{R,i}(\bfx)\mid x_i}\big)^2} &= \big(h_{R,i}^{(0)}\big)^2 + \sum_{m\geq1} \big(a^{(m)}(R_i)\big)^2 .
\end{align*}
Summing over $i$ gives
\begin{align*}
    \tI(\gamma, p_R) = \sum_{i=1}^d \EE{\big(\EE{h_{R,i}(\bfx)\mid x_i}\big)^2} = \sum_{i=1}^d \big(h_{R,i}^{(0)}\big)^2 + \sum_{m\geq 1} \sum_{i=1}^d \big(a^{(m)}(R_i)\big)^2.
\end{align*}
Note that 
\begin{align*}
    \sum_{i=1}^d \big(h_{R,i}^{(0)}\big)^2 = \big\| \EE[\gamma]{h_{R}(\bfx)} \big\|_2^2 = \big\|R \EE[\gamma]{h(R\tran\bfx)}\big\|_2^2 = \big\|\EE[\gamma]{h(\bfx)} \big\|_2^2=:I^{(1)},
\end{align*}
which does not depend on $R$. This proves that
\begin{align*}
    \tI(\gamma, p_R) = I^{(1)} + \sum_{m\geq 1} \sum_{i=1}^d \big(a^{(m)}(R_i)\big)^2.
\end{align*}
It remains to show that $a^{(m)}(\theta)= \langle \calA^{(m)}, \theta^{\otimes m}\rangle$ for a unit vector $\theta$. Using the addition formula of Hermite polynomials, we have
\begin{align*}
    a^{(m)}(\theta)&= \EE{\theta\tran h^{(m)}(\bfx) \he_m(\theta\tran\bfx)}\\
    &=\sum_{i=1}^d \theta_i \sum_{\bfk\in\bbN^d} c_{i,\bfk} \sum_{|\bfj|=m}\sqrt{\frac{m!}{\bfj!} } \theta^{\bfj} \cdot \EE[\gamma]{\he_{\bfk}(\bfx)\cdot \he_{\bfj}(\bfx) }\\
    &=\sum_{i=1}^d \sum_{\bfk: |\bfk|=m} \sqrt{\frac{m!}{\bfk!} }  c_{i,\bfk} \theta^{\bfk+\bfe_i}\\
    &=\sum_{\bfk: |\bfk|=m+1}\sum_{i=1}^d \sqrt{\frac{m!}{(\bfk-\bfe_i)!} } c_{i,\bfk-\bfe_i} \cdot \theta^{\bfk},
\end{align*}
where we use the convention that $c_{i,\bfk-\bfe_i}=0$ if $k_i=0$.
Let $A^{(m+1)}$ be a vector indexed by $\bfk\in\bbN^d$ with $|\bfk|=m+1$ defined as
\begin{align*}
    A^{(m+1)}_{\bfk} =\sum_{i=1}^d \sqrt{\frac{m!}{(\bfk-\bfe_i)!} } \cdot c_{i,\bfk-\bfe_i}.
\end{align*}
Let $\calA^{(m+1)}$ be the symmetric $(m+1)$-order tensor defined as 
\begin{align*}
    \calA^{(m+1)}_{j_1,\ldots,j_{m+1}} = \frac{\bfk!}{(m+1)!} A^{(m+1)}_{\bfk},
    \quad \text{where } \bfk = \sum_{l=1}^{m+1} \bfe_{j_l}.
\end{align*}
Then we can write 
\begin{align*}
    a^{(m)}(\theta)=\sum_{\bfk:|\bfk|=m+1} A^{(m+1)}_{\bfk} \theta^{\bfk} = \sum_{j_1,\ldots,j_{m+1}=1}^d \calA^{(m+1)}_{j_1,\ldots,j_{m+1}} \theta_{j_1}\cdots \theta_{j_{m+1}} = \langle \calA^{(m+1)}, \theta^{\otimes (m+1)}\rangle.
\end{align*}
Note that $h_i(\bfx)=\partial_{x_i} r(\bfx)$ where $r(\bfx)=\log(p/\gamma)(\bfx)$. Since $ \he_{\bfk}(\bfx) =\frac{(-1)^{|\bfk|}}{\sqrt{\bfk!} } \frac{1}{\gamma(\bfx)} \partial^{\bfk} \gamma(\bfx)$, we have
\begin{align*}
    c_{i,\bfk-\bfe_i}&=\EE[\gamma]{h_i(\bfx) \he_{\bfk-\bfe_i}(\bfx)} = \EE[\gamma]{\partial_{x_i} r(\bfx) \he_{\bfk-\bfe_i}(\bfx)} \\
    &=\int \partial_{x_i} r(\bfx)\cdot (-1)^{|\bfk| - 1} \frac{1}{\sqrt{(\bfk-\bfe_i)!}}  \partial^{\bfk-\bfe_i} \gamma(\bfx)  \rd\bfx\\
    &=(-1)^{|\bfk|} \frac{\sqrt{k_i}}{\sqrt{\bfk!}} \int r(\bfx) \partial^{\bfk}\gamma(\bfx)\rd \bfx\\
    &=\sqrt{k_i} \EE[\gamma]{r(\bfx) \he_{\bfk}(\bfx)}.
\end{align*}
Therefore,
\begin{align*}
    A^{(m+1)}_{\bfk}&=\sum_{i=1}^d \frac{\sqrt{m!}}{\sqrt{(\bfk-\bfe_i)!}} \sqrt{k_i}\EE[\gamma]{r(\bfx)\he_{\bfk}(\bfx)}=\sum_{i=1}^d \frac{\sqrt{m!}}{\sqrt{\bfk!}} k_i \EE[\gamma]{r(\bfx)\he_{\bfk}(\bfx)}\\
    &=(m+1) \frac{\sqrt{m!}}{\sqrt{\bfk!}} \EE[\gamma]{r(\bfx)\he_{\bfk}(\bfx)}.
\end{align*}
Therefore,
\begin{align*}
    \calA^{(m+1)}_{j_1,\ldots,j_{m+1}} &= \frac{\bfk!}{(m+1)!} A^{(m+1)}_{\bfk} = \frac{\sqrt{\bfk!}}{\sqrt{m!}} \EE[\gamma]{r(\bfx)\he_{\bfk}(\bfx)}.
\end{align*}
Replacing $m+1$ by $m$ gives the stated formula for $I^{(m)}$.

Furthermore, if $r$ is $m+1$ times differentiable and the derivatives are square-integrable, we have
\begin{align*}
\EE[\gamma]{r(\bfx) \he_{\bfk}(\bfx)} &=\frac{1}{\sqrt{\bfk!}} \int r(\bfx) (-1)^{|\bfk|}\partial^{\bfk}\gamma(\bfx) \rd\bfx  = \frac{1}{\sqrt{\bfk!}}\EE[\gamma]{\partial^{\bfk} r(\bfx)},
\end{align*}
by integration by parts. Hence,
\begin{align*}
    \calA^{(m+1)}_{j_1,\ldots,j_{m+1}} = \frac{1}{\sqrt{m!}} \EE[\gamma]{\partial^{\bfk} r(\bfx)}
\end{align*}

\end{proof}

\subsection{Proof of Proposition~\ref{prop: rotated product exact recovery}}
\label{prf: rotated product exact recovery}

\begin{proof}
First note that
\begin{align*}
    H(\tp) = \EE[\gamma]{\bfx (\nabla\log (\tp/\gamma)(\bfx) )\tran }
\end{align*}
is a diagonal matrix because $\tp$ is a product distribution. Moreover, $H(p) = H(Q\# \tp) = Q H(\tp) Q\tran $. Since the diagonal entries of $H(\tp)$ are distinct by assumption, the eigenvectors of $H(p)$ are exactly the columns of $Q$, up to permutation and sign. Therefore, relative score PCA recovers the correct rotation.
\end{proof}

\subsection{Maximizing higher-degree contribution}
\label{sec: tensor pca}

We describe a simple and popular heuristic based on tensor unfolding for approximately solving the problem 
\begin{align}\label{equ: tensor pca}
\max_{R\in O(d)} \sum_{i=1}^d \langle\calA, R_i^{\otimes m} \rangle^2 
\end{align}
for a symmetric order-$m$ tensor $\calA$. 
We then show that, for the tensor $\calA^{(m)}$ in Theorem~\ref{thm: pFI decomposition}, this heuristic can be implemented efficiently without forming the full tensor. 
More advanced tensor decomposition algorithms are available~\citep{montanari2014statistical,ma2016polynomial,hopkins2016fast}, but pursuing them is beyond the scope of this paper.

The mode-1 unfolding of $\calA\in \otimes^m \R^d$ is the matrix $\calA_{(1)}\in\R^{d\times d^{m-1}}$ whose $i$-th row is the vectorization of the slice $\calA_{i,\cdot,\ldots,\cdot}$. A heuristic solution to \eqref{equ: tensor pca} is given by the top eigenvectors of the matrix $M=\calA_{(1)}\calA_{(1)}\tran$. 
The following proposition shows that this heuristic is exact for orthogonally decomposable (odeco) tensors and that $M$ can be estimated from random contractions of $\calA$, avoiding the explicit construction of $\calA_{(1)}$.
\begin{prop}\label{prop: tensor pca odeco}
    Suppose \(m\ge2\) and
    \(\calA = \sum_{j=1}^r \lambda_j \bfv_j^{\otimes m}\) for some orthonormal
    vectors \(\bfv_j\in\R^d\) and scalars \(\lambda_j\). Then
    \(M=\calA_{(1)}\calA_{(1)}\tran\) has eigenspaces
    \(\operatorname{span}\{\bfv_j:\lambda_j^2=\lambda^2\}\) corresponding to the
    distinct values of \(\lambda^2\). In particular, if the values
    \(\lambda_j^2\) are distinct, the top \(r\) eigenvectors of \(M\) coincide
    with \(\{\bfv_j\}_{j=1}^r\), up to signs and permutations, and the
    corresponding eigenvalues are \(\lambda_j^2\).
    
    Moreover, let $g(\bfa)=\calA(\cdot,\bfa,\ldots,\bfa)=\calA_{(1)} \mathrm{vec}(\bfa^{\otimes (m-1)})$. Then $M = c\cdot \Cov[\bfa\sim\gamma]{g(\bfa)}$ for some constant $c>0$.
\end{prop} 
\begin{proof}
First, the mode-1 unfolding of $\calA$ can be written as
\begin{align*}
    \calA_{(1)} = \sum_{j=1}^r \lambda_j \bfv_j \cdot \mathrm{vec}(\bfv_j^{\otimes (m-1)})\tran.
\end{align*}
Since $\bfv_j$ are orthonormal, we have
\begin{align*}
    M=\calA_{(1)}\calA_{(1)}\tran = \sum_{j=1}^r \lambda_j^2 \bfv_j \bfv_j\tran.
\end{align*}
This proves the first claim.

For the second claim, note that 
\begin{align*}
g(\bfa) = \sum_{j=1}^r \lambda_j \bfv_j\cdot(\bfv_j^{\otimes(m-1)})\tran \mathrm{vec}(\bfa^{\otimes(m-1)})= \sum_{j=1}^r \lambda_j (\bfv_j\tran \bfa)^{m-1} \bfv_j. 
\end{align*}
Denote $z_j=\bfv_j\tran\bfa$. Then
\begin{align*}
    \Cov{g(\bfa)} &= \EE{g(\bfa) g(\bfa) \tran }-\EE{g(\bfa)} \EE{g(\bfa)}\tran \\
    &=\sum_{k,\ell=1}^r \lambda_k \lambda_\ell \EE{z_k^{m-1} z_\ell^{m-1}} \bfv_k \bfv_\ell\tran - \sum_{k,\ell=1}^r \lambda_k \lambda_\ell \EE{z_k^{m-1}} \EE{z_\ell^{m-1}} \bfv_k \bfv_\ell\tran \\
    &=\sum_{k,\ell=1}^r \lambda_k \lambda_\ell \Cov{z_k^{m-1}, z_\ell^{m-1}} \bfv_k \bfv_\ell\tran.
\end{align*}
Since $z_1,\ldots,z_r\iid\N(0,1)$,
$\Cov{z_k^{m-1}, z_\ell^{m-1}}=0$ for $k\neq \ell$. 
Therefore,
\begin{align*}
    \Cov{g(\bfa)} = \Var{z_1^{m-1}} \cdot \sum_{k=1}^r \lambda_k^2\bfv_k \bfv_k\tran.
\end{align*}
Therefore, $M=c\cdot \Cov{g(\bfa)}$ for $c=1/\Var{z^{m-1}}$ where $z\sim\N(0,1)$.
\end{proof}

Proposition~\ref{prop: tensor pca odeco} suggests an efficient randomized heuristic: for independent draws $\bfa\sim\gamma$, compute the mode-1 contractions $\calA(\cdot,\bfa,\ldots,\bfa)$ and use the principal components of the resulting vectors to approximately solve the tensor PCA problem.

\paragraph{Computing the mode-1 contraction for $\calA^{(m)}$}
We now describe how to compute the mode-1 contraction for $\calA^{(m)}$ in Theorem~\ref{thm: pFI decomposition} without explicitly constructing the full tensor $\calA^{(m)}$. 

For a tuple $\boldsymbol{\ell}=(\ell_1,\ldots,\ell_m)\in[d]^m$, let
$\mathrm{cnt}(\boldsymbol{\ell})\in\bbN^d$ denote the vector of index counts:
\begin{align*}
    \mathrm{cnt}(\boldsymbol{\ell}) = (k_1,\ldots,k_d),\text{ where } k_i = \mathrm{card}\big\{k\in[m] : \ell_k=i \big\}.
\end{align*}
The entries of $\calA^{(m)}$ can be rewritten as
\begin{align*}
    \calA^{(m)}_{\ell_1,\ldots,\ell_m} = \frac{\sqrt{\bfk!}}{\sqrt{(m-1)!}} \EE[\gamma]{r(\bfx) \he_{\bfk}(\bfx)},\;\;\text{where}\;\; \bfk=\mathrm{cnt}(\boldsymbol{\ell}).
\end{align*}

The following result shows that $g(\bfa)=\calA^{(m)}(\cdot,\bfa,\ldots,\bfa)$ can be expressed as an expectation involving $h(\bfx)$ and $\he_{m-1}(\bfa\tran\bfx)$. Hence it can be estimated by sampling from $\gamma$, without explicitly constructing the tensor $\calA^{(m)}$.
\begin{prop}\label{pro: contraction of A}
    For any unit vector $\theta\in\R^d$ and $m\geq 2$,
    \begin{align*}
        g(\theta) = \EE[\gamma]{h(\bfx) \he_{m-1}(\theta\tran\bfx)}.
    \end{align*}
    Consequently, for any nonzero vector $\bfa\in\R^d$, 
    \begin{align*}
        g(\bfa)= \|\bfa\|^{m-1} g(\frac{\bfa}{\|\bfa\|}) = \|\bfa\|^{m-1}\EE[\gamma]{h(\bfx)\he_{m-1}(\frac{\bfa\tran\bfx}{\|\bfa\|} )}.
    \end{align*}
\end{prop}
\begin{proof}
By the addition formula for Hermite polynomials, 
\begin{align*}
    \EE[\gamma]{h_i(\bfx) \he_{m-1}(\theta\tran\bfx) } &= \EE[\gamma]{h_i(\bfx) \sum_{\bfj:|\bfj|=m-1} \sqrt{\frac{(m-1)!}{\bfj!}} \theta^{\bfj} \he_{\bfj}(\bfx) }\\
    &=\sum_{\bfj:|\bfj|=m-1} \sqrt{\frac{(m-1)!}{\bfj!}} \theta^{\bfj}\cdot \EE[\gamma]{\partial_{x_i} r(\bfx) \he_{\bfj}(\bfx)}\\
    &=\sum_{\bfj:|\bfj|=m-1} \sqrt{\frac{(m-1)!}{\bfj!}} \theta^{\bfj} \cdot \sqrt{j_i+1} \EE[\gamma]{r(\bfx) \he_{\bfj+\bfe_i}(\bfx)},
\end{align*}
where the last equality follows from $\EE[\gamma]{\partial_{x_i}r(\bfx) \he_{\bfj}(\bfx) } = \sqrt{j_i+1}\EE[\gamma]{r(\bfx) \he_{\bfj+\bfe_i}(\bfx)} $. 
Setting $\bfk=\bfj+\bfe_i$ in the last display, we get
\begin{align*}
    \EE[\gamma]{h_i(\bfx) \he_{m-1}(\theta\tran\bfx) } 
    &=\sum_{\bfk:|\bfk|=m} \sqrt{\frac{(m-1)!}{(\bfk-\bfe_i)! }} \theta^{\bfk-\bfe_i} \sqrt{k_i} \EE[\gamma]{r(\bfx) \he_{\bfk}(\bfx)}\\
    &=\sum_{\bfk:|\bfk|=m} \sqrt{\frac{(m-1)!}{\bfk! }} \theta^{\bfk-\bfe_i} k_i \EE[\gamma]{r(\bfx) \he_{\bfk}(\bfx)}.
\end{align*}

For every $\bfk\in\bbN^d$ with $|\bfk|=m$, there are $\frac{(m-1)!}{(\bfk-\bfe_i)!} $ tuples $\boldsymbol{\ell}=(\ell_1,\ldots,\ell_m)\in[d]^m$ such that $\ell_1=i$ and $\mathrm{cnt}(\boldsymbol{\ell})=\bfk$. Therefore,
\begin{align*}
    \EE[\gamma]{h_i(\bfx) \he_{m-1}(\theta\tran\bfx) } 
    &=\sum_{\bfk:|\bfk|=m}\;  \sum_{\boldsymbol{\ell}:\ell_1=i, \mathrm{cnt}(\boldsymbol{\ell})=\bfk} \frac{(\bfk-\bfe_i)!}{(m-1)!}\sqrt{\frac{(m-1)!}{\bfk! }} \theta^{\bfk-\bfe_i} k_i \EE[\gamma]{r(\bfx) \he_{\bfk}(\bfx)}.
\end{align*}
In the summand, we have $\theta^{\bfk-\bfe_i}=\theta_{\ell_2}\cdots\theta_{\ell_m}$ since $\ell_1=i$. 
Substituting into the last display and plugging in the definition of $\calA^{(m)}$, we get
\begin{align*}
    \EE[\gamma]{h_i(\bfx) \he_{m-1}(\theta\tran\bfx) } 
    &=\sum_{\bfk:|\bfk|=m} \; \sum_{\boldsymbol{\ell}:\ell_1=i, \mathrm{cnt}(\boldsymbol{\ell})=\bfk} \frac{(\bfk-\bfe_i)!}{(m-1)!} \frac{(m-1)!}{\bfk!}k_i \calA_{\ell_1,\ldots,\ell_m} \theta_{\ell_2}\ldots\theta_{\ell_m}\\
    &=\sum_{\ell_2,\ldots,\ell_m=1}^d \calA_{i,\ell_2,\ldots,\ell_m} \theta_{\ell_2}\cdots \theta_{\ell_m} = g(\theta)_i.
\end{align*}
Therefore, $g(\theta)_i=\EE[\gamma]{h_i(\bfx) \he_{m-1}(\theta\tran\bfx)}$.
\end{proof}

We summarize the procedure for approximately maximizing the degree-$m$ contribution in Algorithm~\ref{algo: tensor pca}.
\begin{algorithm}
\caption{Approximate maximization of degree-$m$ contribution}
\label{algo: tensor pca}
\begin{algorithmic}
    \REQUIRE{Target distribution $p$; degree $m$; number of PCs $r$; Monte Carlo sample size $N_1,N_2$}
    \STATE{Generate $\bfa_i\iid \N(0,I_d)$ for $1\leq i\leq N_1$}
    \vspace{.5em}
    \STATE{For each $1\leq i\leq N_1$, compute $g(\bfa_i) = \frac{\|\bfa_i\|^{m-1}}{N_2} \sum_{j=1}^{N_2}  h(\bfx_j) \he_{m-1}(\frac{\bfa_i\tran\bfx_j}{\|\bfa_i\|})$ where $\bfx_j\iid \N(0,I_d)$ for $1\leq j\leq N_2$}
    \vspace{.5em}
    \STATE{Compute the sample covariance $\widehat G$ of $g(\bfa_1),\ldots,g(\bfa_{N_1})$}
    \vspace{.5em}
    \RETURN{Top $r$ eigenvectors of $\widehat G$ corresponding to the largest eigenvalues.}
\end{algorithmic}
\end{algorithm}

\subsection{Derivation of Equation~\eqref{equ: stein discrepancy expression}}
\label{prf: stein discrepancy expression}

Let $\bfu_{\bfa}(\bfx)=\sum_{i=1}^d a_i \he_{m-1}(x_i) \bfe_i$ and
$\bfv_{\bfa}(\bfx)=R\tran \bfu_{\bfa}(R\bfx)=\sum_{i=1}^d a_i R_i\he_{m-1}(R_i\tran\bfx)$. Then we have
\begin{align*}
    \calS^\gamma_{R\# p}(\bfu_\bfa) = \EE[\gamma]{\calT_{R\# p} \bfu_{\bfa}(\bfx)} =\EE[\gamma]{\calT_p \bfv_{\bfa}(\bfx)} = \sum_{i=1}^d a_i \EE{R_i\tran h(\bfx)\cdot \he_{m-1}(R_i\tran\bfx) }.
\end{align*}
By Proposition~\ref{pro: contraction of A}, for a unit vector $\theta$, we have $\langle\calA^{(m)},\theta^{\otimes m}\rangle=\EE[\gamma]{\theta\tran h(\bfx)\cdot \he_{m-1}(\theta\tran\bfx)}$. Therefore,
\begin{align*}
\calS^\gamma_{R\# p}(\bfu_\bfa) = \sum_{i=1}^d a_i \langle \calA^{(m)}, R_i^{\otimes m} \rangle.
\end{align*}

\section{Additional discussion on alternative rotations}
\label{sec: cdr discussion}

In Section~\ref{sec-comparison-to-active-subspace}, we noted that the eigenvectors of the relative Fisher information matrix $\FI(\gamma,p)=\EE[\gamma]{h(\bfx)h(\bfx)\tran}$
do not, in general, recover the correct rotation for rotated product distributions. 
This issue can be addressed by using the centered covariance $\Cov[\gamma]{h(\bfx)}$.
However, even with this fix, there are still examples where it fails to identify a good rotation.

Consider the following 2-dimensional example:
\begin{align}\label{equ: pure interaction example}
r(\bfx)= \log \frac{p}{\gamma}(\bfx)=\beta \sin(x_1)\sin(2x_2) - C,
\end{align}
where $\beta>0$ is a constant and $C$ is the log-normalization constant.
The relative score is
\begin{align*}
    h(\bfx)= \nabla r(\bfx)=\Big(
    \beta \cos(x_1)\sin(2x_2),\,
    2\beta \sin(x_1)\cos(2x_2)
    \Big).
\end{align*}
Since
\begin{align*}
    \EE{\partial_{x_1} r(\bfx)\mid x_1}=0,\qquad 
    \EE{\partial_{x_2} r(\bfx)\mid x_2}=0,
\end{align*}
the projected FI is zero. Thus, in the original coordinates, $\gamma$ is already a stationary point of the MFVI objective, so MFVI will likely get stuck at $\gamma$.

On the other hand, the covariance matrix of the relative score is diagonal. Indeed, the off-diagonal entry of the covariance matrix is
\[
    \Cov{\partial_{x_1}r(\bfx), \partial_{x_2}r(\bfx)}=2\beta^2\Cov{\cos(x_1)\sin(2x_2), \sin(x_1)\cos(2x_2)}=0.
\]
The diagonal entries are $\beta^2\Var{\cos(x_1)\sin(2x_2)}$ and  $4\beta^2\Var{\sin(x_1)\cos(2x_2)}$, which are unequal when $\beta\neq 0$.
Therefore, selecting the rotation using the covariance of the relative score returns the original coordinates. In these coordinates the projected FI is zero, so MFVI will likely get stuck, as shown by the green dashed line in Figure~\ref{fig:pure_interaction_training_loss}.

\begin{figure}
    \centering
    \includegraphics[width=.5\textwidth]{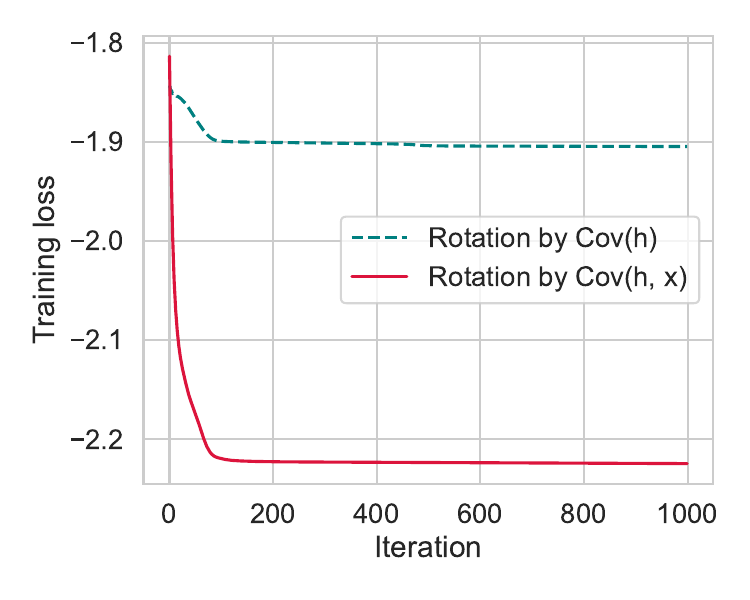}
    \caption{MFVI training loss vs iterations with different choices of rotation. The target distribution is the 2-dimensional distribution defined in \eqref{equ: pure interaction example}.}
    \label{fig:pure_interaction_training_loss}
\end{figure}

Now consider the $H$ matrix:
\[
    H=\Cov{\bfx, h(\bfx)}=
    \begin{pmatrix}
        0 & c\\
        c & 0
    \end{pmatrix}
    \quad\text{where }c=\EE{2\beta\cos(x_1) \cos(2x_2) }.
\]
If we rotate the coordinates by the orthogonal matrix
\[
    Q=\frac{1}{\sqrt{2}}\begin{pmatrix}1&1\\1&-1\end{pmatrix},
\]
the matrix $H$ transforms by conjugation as
\[
    H_{Q} = QH Q\tran
    =
    \begin{pmatrix}
        c & 0\\
        0 & -c
    \end{pmatrix},
\]
which is diagonal with distinct diagonal entries.
Therefore, the proposed relative score PCA method selects the rotation $Q$. 

The projected FI in the rotated coordinates is nonzero, since it is lower bounded by $H_{Q,11}^2+H_{Q,22}^2=2c^2$, which is positive.
Therefore, we expect MFVI in the rotated coordinates to strictly decrease the KL divergence, which is confirmed by the red solid line in Figure~\ref{fig:pure_interaction_training_loss}. 

Figure~\ref{fig:pure_interaction_histograms} compares the samples produced by MFVI under the two rotations. With the rotation selected by relative score PCA, the generated samples closely match the multimodal target distribution. In contrast, with the rotation selected by $\Cov[\gamma]{h(\bfx)}$, the generated samples remain close to the initial standard Gaussian distribution and fail to capture the target.

\begin{figure}
    \centering
    \includegraphics[width=.7\textwidth]{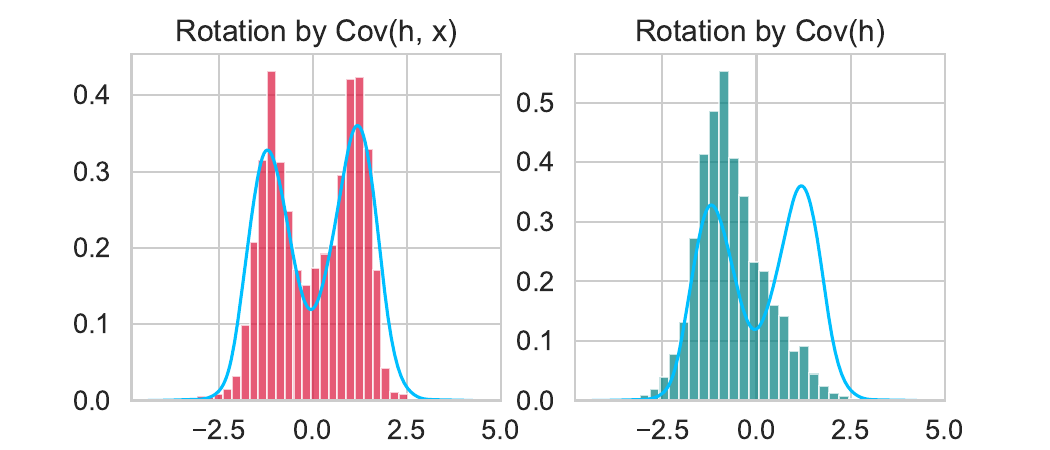}
    \caption{Histograms of MFVI samples projected onto $\frac{1}{\sqrt 2}(1,-1)\tran$ under different rotations. Left: rotation selected by the proposed relative score PCA. Right: rotation selected by the covariance of the relative score. The blue lines represent the marginal densities of the target, obtained by NUTS samples.}
    \label{fig:pure_interaction_histograms}
\end{figure}

\section{Additional experimental results}
\label{sec: numerical details}

\subsection{Additional results for Posteriordb benchmarks}
\label{app: posteriordb}

Table~\ref{tab: posteriordb 2} reports the ESS and KSD for the posteriordb experiments in Section~\ref{sec: posteriordb}.

\begin{table}
    
    \centering
    \begin{tabular}{lcc|cc}
    \toprule
     & \multicolumn{2}{c|}{ESS} & \multicolumn{2}{c}{KSD} \\
    & MF & PCA & MF & PCA \\
    \midrule
    M0 & 1455.5 (270.9) & \textbf{1941.7}\phantom{0} (35.0) & \phantom{0}1.269 (0.078) & \phantom{0}\textbf{0.409} (0.154) \\
    arK & \phantom{00}12.9\phantom{0} (11.5) & \phantom{0}\textbf{257.4} (127.4) & 20.983 (2.837) & \textbf{10.557} (1.075) \\
    garch & \phantom{0}106.3\phantom{0} (94.8) & \phantom{0}\textbf{422.8} (230.9) & \phantom{0}1.611 (0.088) & \phantom{0}\textbf{0.828} (0.104) \\
    gp-regr & \textbf{1931.8}\phantom{0} (15.9) & 1874.7 (150.1) & \phantom{0}\textbf{0.272} (0.079) & \phantom{0}0.286 (0.087) \\
    hmm & \phantom{0}158.0 (122.5) & \textbf{1501.5} (137.9) & \phantom{0}2.962 (0.972) & \phantom{0}\textbf{1.875} (0.785) \\
    kidscore-interaction & \phantom{000}\textbf{9.1}\phantom{00} (7.1) & \phantom{000}7.5\phantom{00} (4.7) & 71.362 (5.679) & \textbf{50.733} (5.515) \\
    mesquite & \phantom{000}7.5\phantom{00} (7.8) & \phantom{00}\textbf{62.6}\phantom{0} (32.8) & \phantom{0}5.212 (0.953) & \phantom{0}\textbf{4.671} (0.809) \\
    nes-logit & \phantom{00}90.3\phantom{0} (97.3) & \textbf{1630.2}\phantom{0} (33.8) & 15.934 (0.700) & \phantom{0}\textbf{2.199} (1.263) \\
    normal-mixture & \phantom{0}680.9 (322.6) & \textbf{1866.3}\phantom{0} (25.9) & 13.926 (0.752) & \phantom{0}\textbf{4.255} (1.738) \\
    radon & 1551.8 (200.4) & \textbf{1939.4}\phantom{0} (15.4) & 22.615 (1.702) & \phantom{0}\textbf{9.175} (2.992) \\
    sesame & \phantom{0}330.6 (166.0) & \textbf{1890.0}\phantom{0} (83.5) & \phantom{0}9.845 (0.609) & \phantom{0}\textbf{2.149} (0.616) \\
    wells & \phantom{00}56.0\phantom{0} (48.3) & \textbf{1610.4}\phantom{0} (42.4) & 18.410 (0.884) & \phantom{0}\textbf{2.971} (0.715) \\
    \bottomrule
    \end{tabular}
\caption{\label{tab: posteriordb 2} Average performance metrics for \texttt{posteriordb} experiments, with standard deviations over 20 independent replicates shown in parentheses. }
\end{table}

\subsection{Poisson GLMM}
\label{sec: glmm details}

We use the first approach introduced in~\citep{tan2021use} to reparametrize the posterior of the Poisson GLMM. The conditional distribution of $b_i\mid \theta_G, \bfy_i$ is proportional to
\begin{align*}
    p(b_i\mid \theta_G,\bfy_i) \propto \varphi(b_i;0,\sigma^2)\cdot \prod_{j=1}^m p(y_{ij} \mid \eta_{ij}), \quad \eta_{ij}=\beta_0 + \beta_1 x_{ij} + b_i.
\end{align*}
The log likelihood $\log p(y_{ij}\mid \eta_{ij}) = y_{ij} \eta_{ij} - e^{\eta_{ij}}$ is approximated by its second-order Taylor expansion at the MLE $\hat\eta^{\text{MLE}}$:
\begin{align*}
    \log p(y_{ij}\mid \eta_{ij}) \approx -\frac12 h_{ij} (\eta_{ij} - \hat\eta^{\text{MLE}}_{ij} )^2,\quad h_{ij}= e^{\hat\eta^{\text{MLE}}_{ij}},
\end{align*}
and $\log p(b_i\mid\theta_G, \bfy_i)$ is approximated by the quadratic function
\begin{align*}
    -\frac{1}{2\sigma^2}b_i^2 - \frac12 \sum_{j=1}^m h_{ij}(\beta_0 + \beta_1 x_{ij} + b_i - \hat\eta^{\text{MLE}}_{ij} )^2.
\end{align*}
The reparametrization $\tb_i=(b_i - \lambda_i)/L_i$ is determined by this Gaussian approximation. When $y_{ij}=0$, we follow~\citep{tan2021use} and define the MLE as $\hat\eta^{\text{MLE}}_{ij} = \psi(y_{ij}+0.5)$, where $\psi$ is the digamma function.

A mean-field Gaussian approximation is first fitted to the reparametrized posterior, and its mean and standard deviation are used to center the target. This initialization step runs for 200 Adam iterations with a learning rate of 0.1. Both Gaussian VI and MFVI+PCA are then optimized for 100 iterations using the same learning rate 0.1. Gradients are computed using 5000 Monte Carlo samples.

Table~\ref{table: glmm local} reports the MSEs in estimating the posterior means and posterior standard deviations for the local effects $b_i$.
\begin{table}
    \centering
    \begin{tabular}{lcc}
    \toprule
     & Gaussian VI & MFVI+PCA \\
    \midrule
    Mean & 0.0007 & 0.0004 \\
    Standard deviation & 0.0043 & 0.0032 \\
    \bottomrule
    \end{tabular}
    \caption{\label{table: glmm local} MSE in estimating the posterior mean and posterior standard deviations for the local effects $b_i$ in Poisson GLMM, averaged over $1\leq i\leq n$.}
\end{table}

\bibliographystyle{abbrvnat}
\bibliography{ref}

\begin{thebibliography}{59}
\providecommand{\natexlab}[1]{#1}
\providecommand{\url}[1]{\texttt{#1}}
\expandafter\ifx\csname urlstyle\endcsname\relax
  \providecommand{\doi}[1]{doi: #1}\else
  \providecommand{\doi}{doi: \begingroup \urlstyle{rm}\Url}\fi

\bibitem[Agrawal and Domke(2025)]{agrawal2025disentangling}
A.~Agrawal and J.~Domke.
\newblock Disentangling impact of capacity, objective, batchsize, estimators, and step-size on flow {VI}.
\newblock In \emph{International Conference on Artificial Intelligence and Statistics}, pages 325--333. PMLR, 2025.

\bibitem[Arnese and Lacker(2024)]{arnese2024convergence}
M.~Arnese and D.~Lacker.
\newblock Convergence of coordinate ascent variational inference for log-concave measures via optimal transport.
\newblock \emph{arXiv preprint arXiv:2404.08792}, 2024.

\bibitem[Balasubramanian et~al.(2022)Balasubramanian, Chewi, Erdogdu, Salim, and Zhang]{balasubramanian2022towards}
K.~Balasubramanian, S.~Chewi, M.~A. Erdogdu, A.~Salim, and S.~Zhang.
\newblock Towards a theory of non-log-concave sampling: first-order stationarity guarantees for {L}angevin {M}onte {C}arlo.
\newblock In \emph{Conference on Learning Theory}, pages 2896--2923. PMLR, 2022.

\bibitem[Bhattacharya et~al.(2025)Bhattacharya, Pati, and Yang]{bhattacharya2025convergence}
A.~Bhattacharya, D.~Pati, and Y.~Yang.
\newblock On the convergence of coordinate ascent variational inference.
\newblock \emph{The Annals of Statistics}, 53\penalty0 (3):\penalty0 929--962, 2025.

\bibitem[Bickel et~al.(2013)Bickel, Choi, Chang, and Zhang]{bickel2013asymptotic}
P.~Bickel, D.~Choi, X.~Chang, and H.~Zhang.
\newblock Asymptotic normality of maximum likelihood and its variational approximation for stochastic blockmodels.
\newblock \emph{The Annals of Statistics}, 41\penalty0 (4), 2013.

\bibitem[Bishop and Nasrabadi(2006)]{bishop2006pattern}
C.~M. Bishop and N.~M. Nasrabadi.
\newblock \emph{Pattern Recognition and Machine Learning}, volume~4.
\newblock Springer, 2006.

\bibitem[Blei et~al.(2017)Blei, Kucukelbir, and McAuliffe]{blei2017variational}
D.~M. Blei, A.~Kucukelbir, and J.~D. McAuliffe.
\newblock Variational inference: A review for statisticians.
\newblock \emph{Journal of the American Statistical Association}, 112\penalty0 (518):\penalty0 859--877, 2017.

\bibitem[Blessing et~al.(2024)Blessing, Jia, Esslinger, Vargas, and Neumann]{blessing2024beyond}
D.~Blessing, X.~Jia, J.~Esslinger, F.~Vargas, and G.~Neumann.
\newblock Beyond {ELBO}s: a large-scale evaluation of variational methods for sampling.
\newblock In \emph{Proceedings of the 41st International Conference on Machine Learning}, pages 4205--4229, 2024.

\bibitem[Brennan et~al.(2020)Brennan, Bigoni, Zahm, Spantini, and Marzouk]{brennan2020greedy}
M.~Brennan, D.~Bigoni, O.~Zahm, A.~Spantini, and Y.~Marzouk.
\newblock Greedy inference with structure-exploiting lazy maps.
\newblock \emph{Advances in Neural Information Processing Systems}, 33:\penalty0 8330--8342, 2020.

\bibitem[Carpenter et~al.(2017)Carpenter, Gelman, Hoffman, Lee, Goodrich, Betancourt, Brubaker, Guo, Li, and Riddell]{carpenter2017stan}
B.~Carpenter, A.~Gelman, M.~D. Hoffman, D.~Lee, B.~Goodrich, M.~Betancourt, M.~Brubaker, J.~Guo, P.~Li, and A.~Riddell.
\newblock Stan: A probabilistic programming language.
\newblock \emph{Journal of Statistical Software}, 76:\penalty0 1--32, 2017.

\bibitem[Che et~al.(2025)Che, Chen, Huan, Huang, and Wang]{che2025stable}
B.~Che, Y.~Chen, Z.~Huan, D.~Z. Huang, and W.~Wang.
\newblock Stable derivative free {G}aussian mixture variational inference for {B}ayesian inverse problems.
\newblock \emph{SIAM Journal on Scientific Computing}, 47\penalty0 (5):\penalty0 A2583--A2608, 2025.

\bibitem[Chen and Gopinath(2000)]{chen2000gaussianization}
S.~Chen and R.~Gopinath.
\newblock Gaussianization.
\newblock \emph{Advances in Neural Information Processing Systems}, 13, 2000.

\bibitem[Constantine et~al.(2014)Constantine, Dow, and Wang]{constantine2014active}
P.~G. Constantine, E.~Dow, and Q.~Wang.
\newblock Active subspace methods in theory and practice: applications to kriging surfaces.
\newblock \emph{SIAM Journal on Scientific Computing}, 36\penalty0 (4):\penalty0 A1500--A1524, 2014.

\bibitem[Cui et~al.(2014)Cui, Martin, Marzouk, Solonen, and Spantini]{cui2014likelihood}
T.~Cui, J.~Martin, Y.~M. Marzouk, A.~Solonen, and A.~Spantini.
\newblock Likelihood-informed dimension reduction for nonlinear inverse problems.
\newblock \emph{Inverse Problems}, 30\penalty0 (11):\penalty0 114015, 2014.

\bibitem[Dinh et~al.(2017)Dinh, Sohl-Dickstein, and Bengio]{dinh2017density}
L.~Dinh, J.~Sohl-Dickstein, and S.~Bengio.
\newblock Density estimation using {Real NVP}.
\newblock In \emph{International Conference on Learning Representations}, 2017.

\bibitem[Du et~al.(2024)Du, Wang, Zhang, and Zhong]{du2024particle}
Q.~Du, K.~Wang, E.~Zhang, and C.~Zhong.
\newblock A particle algorithm for mean-field variational inference.
\newblock \emph{arXiv preprint arXiv:2412.20385}, 2024.

\bibitem[Durkan et~al.(2019)Durkan, Bekasov, Murray, and Papamakarios]{durkan2019neural}
C.~Durkan, A.~Bekasov, I.~Murray, and G.~Papamakarios.
\newblock Neural spline flows.
\newblock \emph{Advances in Neural Information Processing Systems}, 32, 2019.

\bibitem[El~Moselhy and Marzouk(2012)]{el2012bayesian}
T.~A. El~Moselhy and Y.~M. Marzouk.
\newblock Bayesian inference with optimal maps.
\newblock \emph{Journal of Computational Physics}, 231\penalty0 (23):\penalty0 7815--7850, 2012.

\bibitem[Fong et~al.(2010)Fong, Rue, and Wakefield]{fong2010bayesian}
Y.~Fong, H.~Rue, and J.~Wakefield.
\newblock Bayesian inference for generalized linear mixed models.
\newblock \emph{Biostatistics}, 11\penalty0 (3):\penalty0 397--412, 2010.

\bibitem[Gabri{\'e} et~al.(2022)Gabri{\'e}, Rotskoff, and Vanden-Eijnden]{gabrie2022adaptive}
M.~Gabri{\'e}, G.~M. Rotskoff, and E.~Vanden-Eijnden.
\newblock Adaptive {M}onte {C}arlo augmented with normalizing flows.
\newblock \emph{Proceedings of the National Academy of Sciences}, 119\penalty0 (10):\penalty0 e2109420119, 2022.

\bibitem[Gorham and Mackey(2015)]{gorham2015measuring}
J.~Gorham and L.~Mackey.
\newblock Measuring sample quality with {S}tein's method.
\newblock \emph{Advances in Neural Information Processing Systems}, 28, 2015.

\bibitem[Gradshteyn and Ryzhik(2014)]{gradshteyn2014table}
I.~S. Gradshteyn and I.~M. Ryzhik.
\newblock \emph{Table of Integrals, Series, and Products}.
\newblock Academic Press, 2014.

\bibitem[Hastings(1970)]{hastings1970monte}
W.~K. Hastings.
\newblock Monte {C}arlo sampling methods using {M}arkov chains and their applications.
\newblock \emph{Biometrika}, 57\penalty0 (1):\penalty0 97, 1970.

\bibitem[Hillar and Lim(2013)]{hillar2013most}
C.~J. Hillar and L.-H. Lim.
\newblock Most tensor problems are {NP}-hard.
\newblock \emph{Journal of the ACM (JACM)}, 60\penalty0 (6):\penalty0 1--39, 2013.

\bibitem[Hoffman et~al.(2019)Hoffman, Sountsov, Dillon, Langmore, Tran, and Vasudevan]{hoffman2019neutra}
M.~Hoffman, P.~Sountsov, J.~V. Dillon, I.~Langmore, D.~Tran, and S.~Vasudevan.
\newblock Neutra-lizing bad geometry in {H}amiltonian {M}onte {C}arlo using neural transport.
\newblock \emph{arXiv preprint arXiv:1903.03704}, 2019.

\bibitem[Hoffman and Gelman(2014)]{hoffman2014no}
M.~D. Hoffman and A.~Gelman.
\newblock The {No-U-Turn} sampler: adaptively setting path lengths in {Hamiltonian Monte Carlo}.
\newblock \emph{J. Mach. Learn. Res.}, 15\penalty0 (1):\penalty0 1593--1623, 2014.

\bibitem[Hofmann(1994)]{german_credit}
H.~Hofmann.
\newblock {Statlog (German Credit Data)}.
\newblock UCI Machine Learning Repository, 1994.
\newblock {DOI}: https://doi.org/10.24432/C5NC77.

\bibitem[Hopkins et~al.(2016)Hopkins, Schramm, Shi, and Steurer]{hopkins2016fast}
S.~B. Hopkins, T.~Schramm, J.~Shi, and D.~Steurer.
\newblock Fast spectral algorithms from sum-of-squares proofs: tensor decomposition and planted sparse vectors.
\newblock In \emph{Proceedings of the forty-eighth annual ACM symposium on Theory of Computing}, pages 178--191, 2016.

\bibitem[Jiang et~al.(2025)Jiang, Chewi, and Pooladian]{jiang2025algorithms}
Y.~Jiang, S.~Chewi, and A.-A. Pooladian.
\newblock Algorithms for mean-field variational inference via polyhedral optimization in the {W}asserstein space.
\newblock \emph{Foundations of Computational Mathematics}, pages 1--52, 2025.

\bibitem[Jordan et~al.(1999)Jordan, Ghahramani, Jaakkola, and Saul]{jordan1999introduction}
M.~I. Jordan, Z.~Ghahramani, T.~S. Jaakkola, and L.~K. Saul.
\newblock An introduction to variational methods for graphical models.
\newblock \emph{Machine learning}, 37\penalty0 (2):\penalty0 183--233, 1999.

\bibitem[Kingma(2014)]{kingma2014adam}
D.~P. Kingma.
\newblock Adam: A method for stochastic optimization.
\newblock \emph{arXiv preprint arXiv:1412.6980}, 2014.

\bibitem[Kingma et~al.(2016)Kingma, Salimans, Jozefowicz, Chen, Sutskever, and Welling]{kingma2016improved}
D.~P. Kingma, T.~Salimans, R.~Jozefowicz, X.~Chen, I.~Sutskever, and M.~Welling.
\newblock Improved variational inference with inverse autoregressive flow.
\newblock \emph{Advances in Neural Information Processing Systems}, 29, 2016.

\bibitem[Kotz et~al.(2019)Kotz, Balakrishnan, and Johnson]{kotz2019continuous}
S.~Kotz, N.~Balakrishnan, and N.~L. Johnson.
\newblock \emph{Continuous Multivariate Distributions, Volume 1: Models and Applications}, volume~1.
\newblock John Wiley \& Sons, 2019.

\bibitem[Lacker(2026)]{lacker2023independent}
D.~Lacker.
\newblock Independent projections of diffusions: Gradient flows for variational inference and optimal mean field approximations.
\newblock In \emph{Annales de l'Institut Henri Poincare (B) Probabilites et statistiques}, volume~62, pages 638--666. Institut Henri Poincar{\'e}, 2026.

\bibitem[Lacker et~al.(2024)Lacker, Mukherjee, and Yeung]{lacker2024mean}
D.~Lacker, S.~Mukherjee, and L.~C. Yeung.
\newblock Mean field approximations via log-concavity.
\newblock \emph{International Mathematics Research Notices}, 2024\penalty0 (7):\penalty0 6008--6042, 2024.

\bibitem[Laparra et~al.(2011)Laparra, Camps-Valls, and Malo]{laparra2011iterative}
V.~Laparra, G.~Camps-Valls, and J.~Malo.
\newblock Iterative {G}aussianization: from {ICA} to random rotations.
\newblock \emph{IEEE Transactions on Neural Networks}, 22\penalty0 (4):\penalty0 537--549, 2011.

\bibitem[Lavenant and Zanella(2024)]{lavenant2024convergence}
H.~Lavenant and G.~Zanella.
\newblock Convergence rate of random scan coordinate ascent variational inference under log-concavity.
\newblock \emph{SIAM Journal on Optimization}, 34\penalty0 (4):\penalty0 3750--3761, 2024.

\bibitem[Lin et~al.(2019)Lin, Khan, and Schmidt]{lin2019fast}
W.~Lin, M.~E. Khan, and M.~Schmidt.
\newblock Fast and simple natural-gradient variational inference with mixture of exponential-family approximations.
\newblock In \emph{International Conference on Machine Learning}, pages 3992--4002. PMLR, 2019.

\bibitem[Liu and Wang(2016)]{liu2016stein}
Q.~Liu and D.~Wang.
\newblock Stein variational gradient descent: A general purpose {B}ayesian inference algorithm.
\newblock \emph{Advances in Neural Information Processing Systems}, 29, 2016.

\bibitem[Liu(2024)]{liu2024transport}
S.~Liu.
\newblock Transport quasi-{M}onte {C}arlo.
\newblock \emph{arXiv preprint arXiv:2412.16416}, 2024.

\bibitem[Ma et~al.(2016)Ma, Shi, and Steurer]{ma2016polynomial}
T.~Ma, J.~Shi, and D.~Steurer.
\newblock Polynomial-time tensor decompositions with sum-of-squares.
\newblock In \emph{2016 IEEE 57th Annual Symposium on Foundations of Computer Science (FOCS)}, pages 438--446. IEEE, 2016.

\bibitem[Magnusson et~al.(2025)Magnusson, Torgander, B{\"u}rkner, Zhang, Carpenter, and Vehtari]{posteriordb}
M.~Magnusson, J.~Torgander, P.-C. B{\"u}rkner, L.~Zhang, B.~Carpenter, and A.~Vehtari.
\newblock posteriordb: Testing, benchmarking and developing {B}ayesian inference algorithms.
\newblock In \emph{International Conference on Artificial Intelligence and Statistics}, pages 1198--1206. PMLR, 2025.

\bibitem[Meng et~al.(2020)Meng, Song, Song, and Ermon]{meng2020gaussianization}
C.~Meng, Y.~Song, J.~Song, and S.~Ermon.
\newblock Gaussianization flows.
\newblock In \emph{International Conference on Artificial Intelligence and Statistics}, pages 4336--4345. PMLR, 2020.

\bibitem[Metropolis et~al.(1953)Metropolis, Rosenbluth, Rosenbluth, Teller, and Teller]{metropolis1953equation}
N.~Metropolis, A.~W. Rosenbluth, M.~N. Rosenbluth, A.~H. Teller, and E.~Teller.
\newblock Equation of state calculations by fast computing machines.
\newblock \emph{The Journal of Chemical Physics}, 21\penalty0 (6):\penalty0 1087--1092, 1953.

\bibitem[Montanari and Richard(2014)]{montanari2014statistical}
A.~Montanari and E.~Richard.
\newblock A statistical model for tensor {PCA}.
\newblock \emph{Advances in Neural Information Processing Systems}, 27, 2014.

\bibitem[Neville et~al.(2014)Neville, Ormerod, and Wand]{neville2014mean}
S.~E. Neville, J.~T. Ormerod, and M.~Wand.
\newblock Mean field variational {B}ayes for continuous sparse signal shrinkage: Pitfalls and remedies.
\newblock \emph{Electronic Journal of Statistics}, 8:\penalty0 1113--1151, 2014.

\bibitem[Opper and Archambeau(2009)]{opper2009variational}
M.~Opper and C.~Archambeau.
\newblock The variational {G}aussian approximation revisited.
\newblock \emph{Neural computation}, 21\penalty0 (3):\penalty0 786--792, 2009.

\bibitem[Papamakarios et~al.(2021)Papamakarios, Nalisnick, Rezende, Mohamed, and Lakshminarayanan]{papamakarios2021normalizing}
G.~Papamakarios, E.~Nalisnick, D.~J. Rezende, S.~Mohamed, and B.~Lakshminarayanan.
\newblock Normalizing flows for probabilistic modeling and inference.
\newblock \emph{Journal of Machine Learning Research}, 22\penalty0 (57):\penalty0 1--64, 2021.

\bibitem[Papaspiliopoulos et~al.(2007)Papaspiliopoulos, Roberts, and Sk{\"o}ld]{papaspiliopoulos2007general}
O.~Papaspiliopoulos, G.~O. Roberts, and M.~Sk{\"o}ld.
\newblock A general framework for the parametrization of hierarchical models.
\newblock \emph{Statistical Science}, pages 59--73, 2007.

\bibitem[Parno and Marzouk(2018)]{parno2018transport}
M.~D. Parno and Y.~M. Marzouk.
\newblock Transport map accelerated {M}arkov chain {M}onte {C}arlo.
\newblock \emph{SIAM/ASA Journal on Uncertainty Quantification}, 6\penalty0 (2):\penalty0 645--682, 2018.

\bibitem[Rezende and Mohamed(2015)]{rezende2015variational}
D.~Rezende and S.~Mohamed.
\newblock Variational inference with normalizing flows.
\newblock In \emph{International conference on machine learning}, pages 1530--1538. PMLR, 2015.

\bibitem[Sheng et~al.(2025)Sheng, Wu, and Gonz{\'a}lez-Sanz]{sheng2025mode}
S.~Sheng, B.~Wu, and A.~Gonz{\'a}lez-Sanz.
\newblock Mode collapse of mean-field variational inference.
\newblock \emph{arXiv preprint arXiv:2510.17063}, 2025.

\bibitem[{Stan Development Team}(2011)]{stan_irt}
{Stan Development Team}.
\newblock \emph{Stan User's Guide}, 2011.
\newblock \url{https://mc-stan.org/docs/stan-users-guide/regression.html#multilevel-2pl-model}.

\bibitem[Tabak and Vanden-Eijnden(2010)]{tabak2010density}
E.~G. Tabak and E.~Vanden-Eijnden.
\newblock Density estimation by dual ascent of the log-likelihood.
\newblock \emph{Communications in Mathematical Sciences}, 8\penalty0 (1):\penalty0 217--233, 2010.

\bibitem[Tan(2021)]{tan2021use}
L.~S. Tan.
\newblock Use of model reparametrization to improve variational {B}ayes.
\newblock \emph{Journal of the Royal Statistical Society Series B: Statistical Methodology}, 83\penalty0 (1):\penalty0 30--57, 2021.

\bibitem[Tran et~al.(2023)Tran, Tseng, and Kohn]{tran2023particle}
M.-N. Tran, P.~Tseng, and R.~Kohn.
\newblock Particle mean field variational {B}ayes.
\newblock \emph{arXiv preprint arXiv:2303.13930}, 2023.

\bibitem[Zahm et~al.(2022)Zahm, Cui, Law, Spantini, and Marzouk]{zahm2022certified}
O.~Zahm, T.~Cui, K.~Law, A.~Spantini, and Y.~Marzouk.
\newblock Certified dimension reduction in nonlinear {B}ayesian inverse problems.
\newblock \emph{Mathematics of Computation}, 91\penalty0 (336):\penalty0 1789--1835, 2022.

\bibitem[Zhang and Gao(2020)]{zhang2020convergence}
F.~Zhang and C.~Gao.
\newblock Convergence rates of variational posterior distributions.
\newblock \emph{The Annals of Statistics}, 48\penalty0 (4):\penalty0 2180--2207, 2020.

\bibitem[Zhang and Yang(2024)]{zhang2024bayesian}
Y.~Zhang and Y.~Yang.
\newblock Bayesian model selection via mean-field variational approximation.
\newblock \emph{Journal of the Royal Statistical Society Series B: Statistical Methodology}, 86\penalty0 (3):\penalty0 742--770, 2024.

\end{thebibliography}
\end{document}